\documentclass[12pt]{amsart}

\usepackage{courier}

\usepackage{graphicx}
\usepackage{amsmath,amssymb}
\usepackage{amscd}
\usepackage{verbatim}
\usepackage{enumerate}

\usepackage[margin=2cm]{geometry}

\newtheorem{theorem}{Theorem}

\newtheorem{proposition}{Proposition}
\newtheorem{lemma}{Lemma}
\newtheorem{remark}{Remark}
\newtheorem{conjecture}{Conjecture}
\newtheorem{corollary}{Corollary}

\newcommand{\g}{\mathfrak{g}}

\newcommand{\RR}{\mathbb{R}}
\newcommand{\CC}{\mathbb{C}}
\newcommand{\ZZ}{\mathbb{Z}}

\newcommand{\cF}{{\mathcal F}}

\newcommand{\pa}{\partial}

\newcommand{\Ga}{\Gamma}

\newcommand{\cA}{\mathcal A}

\newcommand{\Aut}{\operatorname{Aut}}

\newcommand{\sign}{\operatorname{sign}}

\newcommand{\be}{\begin{equation}}
\newcommand{\ee}{\end{equation}}

\newcommand{\cL}{\mathcal{L}}

\newcommand{\io}{\iota}

\begin{document}

\title{Semiclassical quantization of classical field theories.}
\author{Alberto S. Cattaneo}
\address{Institut f\"ur Mathematik,
Universit\"at Z\"urich,
Winterthurerstrasse 190,
CH-8057, Z\"urich, Switzerland}

\email{alberto.cattaneo@math.uzh.ch}

\author{Pavel Mnev}
\address{Instit\"ut f\"ur Mathematik,
Universit\"at Z\"urich,
Winterthurerstrasse 190,
CH-8057, Z\"urich, Switzerland}

\address{ St. Petersburg Department of V.A.Steklov Institute of Mathematics of the Russian Academy of Sciences, Fontanka 27, St. Petersburg, 191023 Russia}

\address{ Chebyshev Laboratory, St. Petersburg State University,
14th Line, 29b, St. Petersburg, 199178 Russia}

\email{pmnev@pdmi.ras.ru}

\author{Nicolai Reshetikhin}
\address{Department of Mathematics,
University of California, Berkeley
California 94305,
USA}

\email{reshetik@math.berkeley.edu}
\date{\today}
\maketitle

\begin{abstract} These lectures are an introduction to formal
semiclassical quantization of classical field theory. First we
develop the Hamiltonian formalism for classical field theories
on space time with boundary. It does not have to be a cylinder
as in the usual Hamiltonian framework. Then we outline formal
semiclassical quantization in the finite dimensional case.
Towards the end we give an example of such a quantization
in the case of Abelian Chern-Simons theory.
\end{abstract}

\tableofcontents

\section*{Introduction}

The goal of these lectures is an introduction to the
formal semiclassical quantization of classical gauge theories.

In high energy physics space time is traditionally treated as a flat Minkowski
manifold without boundary. This is consistent with the fact the characteristic scale
in high energy is so much smaller then any characteristic scale of the Universe.

As one of the main paradigms in quantum field theory,
quantum fields are usually assigned to elementary particles. The corresponding
classical field theories are described by relativistically invariant local
action functionals. The locality of interactions between elementary particles is one
of the key assumptions of a local quantum field theories and of the Standard Model itself.

The path integral formulation of quantum field theory makes it mathematically
very similar to statistical mechanics. It also suggests that in order to
understand the mathematical nature of local quantum field theory it is natural
to extend this notion from Minkowski space time to a space time with boundary. It is definitely natural to do it for
the corresponding classical field theories.

The concept of topological and conformal field theories
on space time manifolds with boundary was advocated in \cite{At}\cite{Seg}.
The renormalizability of local quantum field theory on a space time with boundary was
studied earlier in \cite{Sym}.
Here we develop the gauge fixing approach for space time manifolds
with boundary by adjusting the Faddeev-Popov (FP) framework to this setting.
This gauge fixing approach is a particular case of the more general Batalin-Vilkovisky (BV)
formalism for quantization of gauge theories.
The classical Hamiltonian part of the BV quantization
on space time manifolds with boundary, the BV-BFV formalism, is developed in \cite{CMR}.
In a subsequent publication we will extend it to the quantum level.

The goal of these notes is an overview of the FP framework in the context of
space time manifolds with boundary.
As a first step we present the Hamiltonian structure for such theories.
We focus on the Hamiltonian formalism for first order theories.
Other theories can be treated similarly, see for example \cite{CMR1} and
references therein. In a subsequent publication we will connect this
approach with the BV-BFV program.

In the first section we recall the concept of local quantum field theory as a
functor from the category of space time cobordisms to the category of vector
spaces. The second section contains examples: the scalar field theory,
Yang-Mills theory, Chern-Simons and BF theories. The concept of
semiclassical quantization of first order quantum field theories is explained in section 3
where we present a finite dimensional model for the gauge fixing for
space time manifolds with or without boundary.
In section 4 we briefly discuss the example of Abelian Chern-Simons theory. The nonabelain case and
the details of the gluing of partition functions
semiclassical Chern-Simons theories will be given elsewhere.

The authors benefited from discussions with T. Johnson-Freyd, J. Lott and B. Vertman,
A.S.C. acknowledges partial support of SNF Grant No. 200020-131813/1,
N.R. acknowledges the support from the Chern-Simons grant, from
the NSF grant DMS-1201391 and the University of Amsterdam where important part of the work
was done. P. M. acknowledges partial support of RFBR Grant No.13-01-12405-ofi-m-2013, of the Chebyshev Laboratory  (Department of Mathematics and Mechanics, St. Petersburg State University)  under RF Government grant 11.G34.31.0026, of  JSC ``Gazprom Neft'', and of SNF Grant No. 200021-13759. We also benefited from hospitality and
research environment of the QGM center at Aarhus University.

\section{First order classical field theories}

\subsection{Space time categories}

In order to define a classical field theory one has to specify a
space time category, a space of fields for each space time and the
action functional on the space of fields.

Two space time categories which are most important for these lectures
are the category of smooth $n$-dimensional cobordisms and the category
of smooth $n$-dimensional Riemannian manifolds.

{\bf The $d$-dimensional smooth category}. \index{topological
category} Objects are smooth, compact, oriented $(d-1)$-dimensional
manifolds with smooth $d$-dimensional collars. A morphism between $\Sigma_1$ and $\Sigma_2$ is a smooth
$d$-dimensional compact oriented manifolds
with $\pa M=\Sigma_1\sqcup \overline{\Sigma_2}$ and the smooth structure on $M$ agrees with smooth structure
on collars near the boundary. The orientation on $M$
should agree with the orientations of $\Sigma_1$ and be opposite to the one on $\Sigma_2$ in a natural way.

The composition consists of gluing two morphisms along the common boundary in
such a way that collars with smooth structure on them fit.

{\bf The $d$-dimensional Riemannian category}. \index{Riemannian
category} Objects are $(d-1)$-dimensional Riemannian manifolds with $d$-dimensional collars.  Morphisms
between two oriented $(d-1)$-dimensional Riemannian manifolds $N_1$
and $N_2$ are oriented $d$-dimensional Riemannian manifolds $M$ with collars
near the boundary,
such that $\pa M =N_1\sqcup \overline{N_2}$. The orientation on all
three manifolds should naturally agree, and the metric on $M$ agrees
with the metric on $N_1$ and $N_2$ and on collar near the boundary. The
composition is the gluing of such Riemannian cobordisms. For the
details see \cite{ST}.

This category is important for many reasons. One of them is that
it is the underlying structure for statistical quantum field theories.

{\bf The $d$-dimensional pseudo-Riemannian category} The difference between this category and the Riemannian category is that morphisms
are pseudo-Riemannian with the signature $(d-1,1)$ while objects remain $(d-1)$-dimensional Riemannian.
This is the most interesting category for particle physics.

Both objects and morphisms may have an extra structure such as
a fiber bundle (or a sheaf) over it. In this case such structures
for objects should agree with the structures for morphisms.

\subsection{General structure of first order theories}

\subsubsection{First order classical field theories}
A first order classical field theory\footnote{It is not essential that we consider
here only first order theories. Higher order theories where $L(d\phi, \phi)$ is not
necessary a linear function in $d\phi$ can also be treated in a similar way, see for
example \cite{CMR1} and references therein. In first order theories the space of boundary fields
is the pull-back of fields in the bulk.} is defined by the following data:
\begin{itemize}

\item A choice of space time category.

\item A choice of the space of fields $F_M$ for each space time manifold $M$.
This comes together with the definition of the space of fields $F_{\pa M}$ for the
boundary of the space time and the restriction mapping $\pi: F_M\to F_{\pa M}$.

\item A choice of the action functional on the space $F_M$ which
is local and first order in derivatives of fields, i.e.
\[
S_M(\phi)=\int_M L(d\phi, \phi)
\]
Here $L(d\phi, \phi)$ is linear in $d\phi$.
\end{itemize}

These data define:
\begin{itemize}

\item The space $EL_M$ of solutions of the Euler-Lagrange equations.

\item The $1$-form $\alpha_{\pa M}$ on the space of boundary fields.

\item The Cauchy data subspace $C_{\pa M}$ of boundary values (at $\{0\}\times \pa M$) of solutions of the Euler-Lagrange equations in $[0, \epsilon)\times \pa M$.

\item The subspace $L_M\subset C_{\pa M}$ of boundary values of solutions of the Euler-Lagrange
equations in $M$, $L_M=\pi (EL_M)$.

\end{itemize}

When $C_{\pa M}\neq F_{\pa M}$ the Cauchy problem is overdetermined
and therefore the action is degenerate. Typically it is degenerate because of
the gauge symmetry.

A natural boundary condition for such system is given by a Lagrangian fibration\footnote{Here and throughout the text
``fibration'' will mean either the projection to the space of leaves of a foliation where the fiber
is a leaf, or an actual fiber bundle. We will specify this when needed.} on the space of boundary fields such that the form $\alpha_{\pa M}$
vanishes at the fibers. The last conditions guarantees that
solutions of Euler-Lagrange equations which are constrained to a leaf
of such fibration are critical points of the action functional, i.e. not
only the bulk term vanishes but also the boundary terms.

\subsubsection{First order classical field theory as a functor}

First order classical field theory can be regarded as
a functor from the category of space times to the category
which we will call Euler-Lagrange category and will denote $\underline{EL}$.
Here is an outline of this category:

An {\it object} of $\underline{EL}$ is a symplectic manifold $F$ with a prequantum line bundle\footnote{ The prequantum line bundle can be trivial,
the Hermitian connection $\alpha_F$ in this case is a $1$-form whose coboundary is the
symplectic form.}, i.e a line bundle
with a connection $\alpha_F$, 
such that the symplectic form
is the curvature of this connection. It should also have a Lagrangian
foliation\footnote{The foliation does not have to be a fiber bundle.
Its fiber should be Lagrangian over generic points of the base.} which is $\alpha_F$-exact,
i.e. the pull-back of $\alpha_F$ to each fiber vanishes.

A {\it morphism} between $F_1$ and $F_2$ is a manifold $F$
together with two projections $\pi_1: F\to F_1$ and $\pi_2: F\to F_2$,
with a function $S_F$ on $F$ and with the subspace $EL\subset F$ such that
$dS_F|_{EL}$ is the pull-back of $-\alpha_{F_1}+\alpha_{F_2}$ on $F_1\times F_2$.
The image of $EL$ in $(F_1, -\omega_1)\times (F_2,\omega_2)$ is an isotropic submanifold.
Here $\omega_i=d\alpha_{F_i}$. We will focus on theories where
these subspaces are Lagrangian.

The composition of morphisms $(F, S_F)$ and $(F',S_{F'})$ is the fiber product of the morphism spaces
$F$ and $F'$ over the intermediate object and $S_{F'\circ F}=S_F+S_{F'}$. This category is the $gh=0$ part of the BV-BFV category from \cite{CMR}.

A first order classical field theory defines a
functor from the space time category to the Euler-Lagrange category.
An object $N$ of the space time category is mapped to the space of fields $F_N$,
a morphism $M$ is mapped to $(F_M, S_M)$, etc. Composition of morphisms is mapped
to the fiber product of fields and because of the assumption of locality of the
action functional, it is additive with respect to the gluing.

This is just an outline of the Euler-Lagrange category and of the functor.
For our purpose of constructing formal semiclassical quantization we will not need
the precise details of this construction. But it is important to have this more
general picture in mind.

\subsection{Symmetries in first order classical field theories}

The theory is relativistically invariant if the action is
invariant with respect to geometric automorphisms of the
space time. These are diffeomorphisms for the smooth category, isometries for
the Riemannian category etc. In such theory the action is constructed
using geometric operations such as de Rham differential and
exterior multiplication of forms for smooth category. In Riemannian category in addition
to these two operations we have Hodge star (or the metric).

If the space time category has an additional structure such as fiber bundle, the automorphisms of
this additional structure give additional symmetries of the theory. In
Yang-Mills, Chern-Simons and BF theories, gauge symmetry, or automorphisms
of the corresponding principal $G$-bundle, are such a symmetry. A theory
with such space time with the gauge invariant action is called gauge invariant.
The Yang-Mills theory is gauge invariant, the Chern-Simons and the BF theories are
gauge invariant only up to boundary terms.

There are more complicated symmetries when a distribution, not necessary integrable,
is given on the space of fields and the action is annihilated by
corresponding vector fields. Nonlinear Poisson $\sigma$-model is an example
of such field theory \cite{CF}.

\section{Examples}

\subsection{First order Lagrangian mechanics}
\subsubsection{The action and boundary conditions} In Lagrangian mechanics the main
component which determines the dynamics is the Lagrangian function. This is a function
on the tangent bundle to the configuration space $L(\xi, x)$ where $\xi\in T_xN$.
In Newtonian mechanics  the Lagrangian function is quadratic in velocity
and the quadratic term is positive definite which turns $N$ into a Riemannian manifold.

The most general form of first order Lagrangian is $L(\xi,x)=<\alpha(x), \xi>-H(x)$ where
$\alpha$ is a $1$-form on $N$ and $H$ is a function on $N$.
The action of a first order Lagrangian mechanics is the following functional on parameterized paths
$F_{[t_1,t_2]}=C^\infty([t_1,t_2], N)$
\begin{equation}\label{act-cl-m}
S_{[t_2,t_1]}[\gamma]=\int_{t_1}^{t_2} (\left< \alpha(\gamma(t)),\dot\gamma(t)\right> -H(\gamma(t)))dt,
\end{equation}
where $\gamma$ is a parametrized path.

The Euler-Lagrange equations for this action are:
\[
\omega(\dot\gamma(t))-dH(\gamma(t))=0,
\]
where $\omega=d\alpha$. Naturally, the first order Lagrangian system
\index{first order Lagrangian system} is called {\it
non-degenerate}, if the form $\omega$ is non-degenerate. We will focus on non-degenerate
theories here. Denote the space of solutions
to Euler-Lagrange equations by $EL_{[t_1,t_2]}$.

Thus, a non-degenerate first
order Lagrangian system defines an exact
symplectic structure $\omega=d\alpha$ on a manifold $N$.
The Euler-Lagrange equations
for such system are equations for flow lines of the Hamiltonian on
the symplectic manifold $(N, \omega)$ generated by the Hamiltonian
$H$. It is clear that the action of a non-degenerate first
order system is exactly the action for this
Hamiltonian system.

The variation of the action on solutions of the Euler-Lagrange equations is
given by the boundary terms:
\begin{equation*}
\delta  S_{[t_2,t_1]}[\gamma]=\left< \alpha(\gamma(t)),\delta\gamma(t)\right> |_{t_1}^{t_2}\, .
\end{equation*}
If $\gamma(t_1)$ and $\gamma(t_2)$ are constrained to Lagrangian
submanifolds in $L_{1,2}\subset N$ with
$TL_{1,2}\subset\ker(\alpha)$, these terms vanish.

The restriction to boundary points gives the projection $\pi: F_{[t_1,t_2]}\to N\times N$.
The image of the space of solutions of the Euler-Lagrange equations
$L_{[t_1,t_2]}\subset N\times N$ for small $[t_1, t_2]$ is a
Lagrangian submanifold with respect to the symplectic form $(d\alpha)_1-(d\alpha)_2$
on $N\times N$.

Note that solutions of the Euler-Lagrange equation with boundary conditions
in $L_1\times L_2$ correspond to the intersections points $(L_1\times L_2)\cap L_{[t_1,t_2]}$
which is generically a discrete set.

\subsubsection{More on boundary conditions} The evolution of the system from time $t_1$ to $t_2$ and then to $t_3$
can be regarded as gluing of space times $[t_1,t_2]\times [t_2,t_3]\to [t_1,t_2]\cup [t_2,t_3]=[t_1,t_3]$.
If we impose boundary conditions $L_1,L_2,L_3$ at times $t_1,t_2,t_3$ respectively there may be no continuous
solutions of equations of motion for intervals $[t_1,t_2]$ and $[t_2,t_3]$ which would compose into
a continuous solution for the interval $[t_1,t_3]$. This is why boundary conditions should come in families of
Lagrangian submanifolds, so that by varying the boundary condition at $t_2$ we could choose $L_2$ in such a way that
solutions for $[t_1,t_2]$ and $[t_2,t_3]$ would compose to a continuous solution.

This is why we will say that a boundary condition for a first order theory is a Lagrangian fibration on the space of boundary values of classical fields. In case of first order classical mechanics this is a Lagrangian fibration on $N$, boundary condition is a
Lagrangian fibration of $(N,\omega)\times (N, -\omega)$. It is natural to choose
boundary conditions independently for each connected component of the boundary of
the space time. In case of classical mechanics this means a choice of Lagrangian
fibration $p: N\to B$ for each endpoint of $[t_1,t_2]$.
The form $\alpha$ should vanish on fibers of this fibration.

\begin{remark} For semiclassical quantization we will need only
classical solutions and infinitesimal neighborhood of classical solutions. This means
that we need in this case a Lagrangian fibration on the space of boundary fields
defined only locally, not necessary globally.
\end{remark}

Let $N$ be a configuration space (such as $\mathbb{R}^n$) and $T^{\ast}(N)$ be the corresponding phase space. Let $\gamma$ be a parameterized path in $T^{\ast}(N)$ such that, writing $\gamma(t) = (p(t), q(t))$ (where $p$ is momenta and $q$ is position), we have $q(t_i) = q_i$ for two fixed points $q_1, q_2$. If $\gamma_{\text{cl}}$ is a solution to the Euler-Lagrange equations, then

\begin{equation} d S^{\gamma_{\text{cl}}}_{t_1, t_2}(q_1, q_2) = \pi^{\ast}(p_1 \, dq_1 - p_2 \, dq_2) \end{equation}
where $p_1 = p(t_1), p_2 = p(t_2)$ are determined by $t_1, t_2, q_1, q_2$. The function $S^{\gamma_{\text{cl}}}_{t_1, t_2}$  is the Hamilton-Jacobi function.

\subsection{Scalar field theory in an $n$-dimensional space time} The space time in this case is a smooth oriented compact Riemannian manifold $M$ with $\dim M = n$. The space of fields is

\begin{equation} F_M = \Omega^0(M) \oplus \Omega^{n-1}(M). \end{equation}
where we write $\varphi$ for an element of $\Omega^0(M)$ and $p$ for an element of $\Omega^{n-1}(M)$. The action functional is

\begin{equation} S_M(p, \varphi) = \int_M p \wedge d \varphi - \frac{1}{2} \int_M p \wedge \ast p - \int_M V(\varphi) \, dx. \end{equation}

The first term is topological and analogous to $\int_{\gamma} \alpha$ in (\ref{act-cl-m}). The second and third terms 
use the metric and together yield an analog of the integral of the Hamiltonian in (\ref{act-cl-m}).

The variation of the action is

\begin{equation} \int_M \delta p \wedge (d \varphi - \ast p) - (-1)^{n-1} \int_M dp \wedge \delta \varphi + (-1)^{n-1} \int_{\partial M} p \, \delta \varphi - \int_M V'(\varphi) \, \delta \varphi \, dx. \end{equation}

The Euler-Lagrange equations are therefore

\begin{equation} d \varphi - \ast p = 0, (-1)^{n-1} dp + V'(\varphi) \, dx = 0. \end{equation}

The first equation gives $p = (-1)^{n-1} \ast d \varphi$, and substituting this into the second equation gives

\begin{equation} \Delta \varphi + V'(\varphi) \, dx = 0. \end{equation}
where $\Delta=\ast d \ast d$ is the Laplacian acting of functions.

Thus the space of all solutions of Euler-Lagrange equations is
\[
EL_M=\{(p, \varphi)|p = (-1)^{n-1} \ast d \varphi, \ \ \Delta \varphi + V'(\varphi) = 0 \}
\]

\begin{remark} To recover the second-order Lagrangian
compute the action at the critical point in $p$, i.e. substitute $p = (-1)^{n-1} \ast d \varphi$ into the action functional:

\begin{eqnarray*} S_M((-1)^{n-1} \ast d \varphi, \varphi) &=& \int_M (-1)^{n-1} \ast d \varphi \wedge d \varphi - \frac{1}{2} \int \ast d \varphi \wedge \ast\ast d \varphi\\ &-&   \int_M V(\varphi) \, dx
 = \frac{1}{2} \int_M d \varphi \wedge \ast d \varphi - \int_M V(\varphi) \, dx \\
 &=& \int_M \left( \frac{1}{2} (d \varphi, d \varphi) - V(\varphi) \right) \, dx. \end{eqnarray*}

\end{remark}

The boundary term in the variation gives the $1$-form on boundary fields

\begin{equation} \alpha_{\partial M} = \int_{\partial M} p \, \delta \varphi \in \Omega^1(F_{\partial M}) \end{equation}
Here $\delta$ is the de Rham differential on $\Omega^{\bullet}(F_{\partial M})$.
The differential of this $1$-form gives the symplectic form $\omega_{\pa M}=\delta\alpha_{\pa M}$ on $F_{\pa M}$.

Note that we can think of the space $F_{\partial M}$ of boundary fields as $T^{\ast}(\Omega^0(\partial M))$ in the following manner: if $\delta \varphi \in T_{\varphi}(\Omega^0(\partial M)) \cong \Omega^0(\partial M)$ is a tangent vector,
 then the value of the cotangent vector $A \in \Omega^{n-1}(\partial M)$ is

\begin{equation} A(\delta \varphi) = \int_{\partial M} A \wedge \delta \varphi. \end{equation}
The symplectic form $\omega_{\pa M}$ is the natural symplectic form on $T^*\Omega^0(\partial M)$.

The image of the space $EL_M$ of all solutions to the Euler-Lagrange equations
with respect to the restriction map $\pi : F_M \to F_{\partial M}$ gives a
subspace $L_M=\pi(EL_M)\subset F_{\pa M}$.

\begin{proposition} Suppose there is a unique solution\footnote{It is unique if $-V(\varphi)$ is convex.} to $\Delta \varphi + V'(\varphi) = 0$ for any Dirichlet boundary condition $\varphi|_{\partial M} = \eta$. Then $\pi(EL_M)$ is a Lagrangian submanifold of $F_{\partial M}$. \end{proposition}

Indeed, in this case $L_M$ is the graph of a map $\Omega^0(\partial M) \to F_{\partial M}$
given by $\eta\mapsto (p_\pa=\pi((-1)^{n-1}*d\varphi), \eta)$ where $\varphi$ is the
unique solution to the Dirichlet problem with boundary conditions $\eta$.

The space of boundary fields has a natural Lagrangian  fibration $\pi_\pa:  T^{\ast}(\Omega^0(\partial M)) \to \Omega^0(\partial M)$. This fibration corresponds to
Dirichlet boundary conditions: we fix the value $\varphi|_{\partial M}=\eta$ and impose no conditions on $p|_{\partial M}$, i.e. we impose boundary condition $(p, \varphi)|_{\pa M}\in
\pi_\pa^{-1}(\eta)$.

Another natural family of boundary conditions, Neumann boundary conditions, correspond to the Larganian
fibration of $ T^{\ast}(\Omega^0(\partial M))\simeq \Omega^{n-1}(\pa M)\oplus \Omega^0(\pa M)$
where the base is $ \Omega^{n-1}(\pa M)$. In the case we fix
$\ast_\pa i^*(p)=\eta\in \Omega^0(\pa M)$. The intersection of $L_M$ and the fiber over $\eta$
is the set of pairs $(\ast_\pa\eta, \xi)\in \Omega^{n-1}(\pa M)\oplus \Omega^0(\pa M)$
where $\xi=i^*(\phi)$ and $\phi$ is a solution to the Neumann problem
\[
\Delta\phi+V'(\phi)=0, \ \ \pa_n\phi|_{\pa M}=\eta
\]
where $\pa_n$ is the normal derivative of $\phi$ at the boundary.

\subsection{Classical Yang-Mills theory}

Space time is again a smooth compact oriented Riemannian manifold $M$. Let $G$ be a compact semisimple, connected, simply-connected Lie group with Lie algebra $\mathfrak{g}$. We assume that it is a matrix group, i.e. we fix an embedding of $G$ into $\text{Aut}(V)$, and hence an embedding of $\mathfrak{g}$ into $\text{End}(V)$ such that the
Killing form on $\mathfrak g$ is $<a,b>=tr(ab)$. The space of fields in the first order
Yang-Mills theory is

\begin{equation} F_M = \Omega^1(M, \mathfrak{g}) \oplus \Omega^{n-2}(M, \mathfrak{g}) \end{equation}
where we think of $\Omega^1(M, \mathfrak{g})$ as the space of connections on a trivial $G$-bundle over $M$. If we use a nontrivial $G$-bundle over $M$ then the first term should be replaced by the corresponding space of connections. We denote an element of $F_M$ by an ordered pair $(A, B)$, $A\in \Omega^1(M, \mathfrak{g}) $ and $B\in \Omega^{n-2}(M, \mathfrak{g})$ . The action functional is

\begin{equation} S_M(A, B) = \int_M \text{tr}(B \wedge F(A)) - \frac{1}{2} \int_M \text{tr}(B \wedge \ast B) \end{equation}
where $F(A)=dA+A\wedge A$ is the curvature of $A$ as a connection\footnote{We will use notations $A\wedge B=\sum_{\{i\}\{j\}}A_{\{i\}}B_{\{j\}}dx^{\{i\}}\wedge dx^{\{j\}}$
for matrix-valued forms $A$ and $B$. Here $\{i\}$ is a multiindex $\{i_1,\dots, i_k\}$ and
$x^i$ are local coordinates on $M$. We will also write $[A\wedge B]$ for $\sum_{\{i\}\{j\}}[A_{\{i\}}, B_{\{j\}}] dx^{\{i\}}\wedge dx^{\{j\}}$.}.

After integrating by part we can write the variation of the action as the sum of bulk and boundary parts:

\begin{equation} \delta S_M(A,B)= \int_M \text{tr}(\delta B\wedge (F(A)-*B)+ \delta A\wedge d_A B ) - \int_{\partial M} \text{tr}(\delta A\wedge B ) \end{equation}
The space $EL_M$ of all solution to Euler-Lagrange equations is the space of pairs $(A, B)$
which satisfy
\[
B=*F(A), \ \ d_AB=0
\]

\subsubsection{} The boundary term of the variation defines the one-form on the space boundary fields $F_{\pa M}=\Omega^1(\pa M)\oplus \Omega^{n-2}(\pa M)$.

\begin{equation}  \alpha_{\partial M} = -\text{tr} \int_{\partial M} \delta A\wedge B \in \Omega^1(F_{\pa M}) \end{equation}

Its differential defines the symplectic form $\omega_{\pa M}=\int_{\pa M} \text{tr} (\delta A\wedge \delta B)$.

Note that similarly to the scalar field theory boundary fields can be regarded as $T^*\Omega^1(\pa M)$
where we identify cotangent spaces with $\Omega^{n-2}(\pa M)$, tangent spaces with $\Omega^1(\pa M)$
with the natural pairing
\[
\beta(\alpha)=\text{tr} \int_{\pa M}  \alpha\wedge \beta
\]

The projection map $\pi : F_M \to F_{\partial M}$ which is the restriction (pull-back) of forms
to the boundary defines the subspace $L_M = \pi(EL_M)$ of the space of boundary values of solutions to the Euler-Lagrange equations on $M$.

\subsubsection{} Let us show that this subspace is Lagrangian for Maxwell's electrodynamics, i.e.
for the Abelian
Yang-Mills with $G=\RR$. In this case Euler-Lagrange equations are
\[
B=\ast dA, \ \  \  \ d\ast dA=0
\]
Fix Dirichlet boundary condition $i^*(A)=a$. Let $A_0$ be a solution to this equation satisfying Laurenz gauge
condition $d^*A_0=0$. Such solution is a harmonic $1$-form, $(dd^*+d^*d)A_0=0$ with boundary condition $i^*(A_0)=a$.
If $A_0'$ is another such form, then $A_0-A_0'$ is a harmonic $1$-form with boundary condition $i^*(A_0-A_0')=0$.
The space of such forms is naturally isomorphic to $H^1(M,\pa M)$.
Each of these solutions gives the same value for $B=\ast dA=\ast dA_0$ and therefore
its boundary value $b=i^*(B)$ is uniquely determined by $a$.
Therefore the projection of $EL_M$ to the boundary is a graph of the map $a\to b$
and thus $L_M$ is a Lagrangian submanifold.

The Dirichlet and Neumann boundary value problems for Yang-Mills theory were studied in \cite{Mar}.

\begin{conjecture} The submanifold $L_M$ is Lagrangian for non-Abelian Yang-Mills theory.
\end{conjecture}

It is clear that this is true for small connections, when we can rely
on perturbation theory staring from an Abelian connection. It is also easy to prove that
$L_M$ is isotropic.

\subsubsection{} Define the \emph{Cauchy subspace}

\begin{equation}
C_{\partial M} = \pi_\epsilon(EL_{\partial M_{\epsilon}})
\end{equation}

where $\partial M_{\epsilon} = [0, \epsilon) \times \partial M$ and $\pi_\epsilon: F_{\partial M_{\epsilon}}\to F_{\pa M}$ is the restriction of fields to $\{0\}\times \pa M$. In other words $C_{\pa M}$ is the space of
boundary values of solution to Euler-Lagrange equations in $\partial M_{\epsilon} = [0, \epsilon) \times \partial M$. It is easy to see that\footnote{
The subspace $C_{\partial M}$ also makes sense also in scalar field theory, where explicitly it consists of pairs $(p, \varphi) \in \Omega^{n-1}(\partial M) \oplus \Omega^0(\partial M)$ where $p$ is the pullback of $p_0 = \ast d \varphi_0$ and $\varphi$ is the boundary value of $\varphi_0$ which solves the Euler-Lagrange equation  $\Delta \varphi_0 - V'(\varphi_0) = 0$. Since Cauchy problem has unique solution in a small neighborhood of the
boundary, $C_{\pa M}=F_{\pa M}$ for the scalar field.}
\[
C_{\pa M}=\{(A,B)| d_AB=0\}
\]
We have natural inclusions

\[
L_M\subset C_{\pa M} \subset F_{\pa M}
\]

\subsubsection{} The automorphism group of the trivial principal $G$-bundle over $M$ can be naturally identified with $C^\infty(M,G)$. Bundle automorphisms act on the space of Yang-Mills fields. Thinking of a connection $A$ as an element $A \in \Omega^1(M, \mathfrak{g})$ we have the following formulae for the action of
the bundle automorphism (gauge transformation) $g$ on fields:
\begin{equation}\label{g-conn} g : A \mapsto A^g = g^{-1} A g + g^{-1} dg, \ \ \  B \mapsto B^g = g^{-1} B g. \end{equation}
Note that the curvature $F(A)$ is a 2-form and it transforms as $F(A^g)=g^{-1}F(A)g$. Also, if we have
two connections $A_1$ and $A_2$, their difference is a 1-form and $A_1^g-A_2^g=g^{-1}(A_1-A_1)g$.

The Yang-Mills functional is invariant under this symmetry:

\begin{equation} S_M(A^g, B^g) = S_M(A,B) \end{equation}
which is just the consequence of the cyclic property of the trace.

The restriction to the boundary gives the projection map of gauge groups $\tilde{\pi} : G_M \to G_{\partial M}$ which is a group homomorphism. This map is surjective, so we obtain an exact sequence

\begin{equation} 0 \to \text{Ker}(\tilde{\pi}) \to G_M \to G_{\partial M} \to 0 \end{equation}
where $\text{Ker}(\tilde{\pi})$ is the group of gauge transformations acting trivially at the boundary.

It is easy to check that boundary gauge transformations $G_{\partial M}$ preserve the symplectic form $\omega_{\partial M}$. The action of $G_M$ induces an infinitesimal action of the Lie algebra $\mathfrak{g}_M=C^\infty(M,\mathfrak{g})$ of $G_M$
by vector fields on $F_M$. For $\lambda \in \mathfrak{g}_M$  we denote by $(\delta_{\lambda} A , \delta_{\lambda} B)$ the tangent vector to $F_M$ at the point $(A,B)$ corresponding to the action of $\lambda$:

\begin{equation} \delta_{\lambda} A = -[\lambda, A] + d \lambda = d_A \lambda, \delta_{\lambda} B = -[\lambda, B] \end{equation}
where the bracket is the pointwise commutator (we assume that $\mathfrak{g}$ is a matrix Lie algebra). Recall that the action of a Lie group on a symplectic manifold is Hamiltonian
if vector fields describing the action of the Lie algebra $Lie(G)$ are Hamiltonian.

We have the following

\begin{theorem} The action of $G_{\partial M}$ on $F_{\partial M}$ is Hamiltonian. \end{theorem}

Indeed, let $f$ be a function on $F_{\partial M}$ and let $\lambda \in \mathfrak{g}_{\partial M}$. Let $\delta_{\lambda} f$ denote the  Lie derivative of the corresponding infinitesimal gauge transformation. Then

\begin{equation} \delta_{\lambda} f ((A, B)) = \int_{\partial M} \text{tr} \left( \frac{\delta f}{\delta A} \wedge d_A \lambda + \frac{\delta f}{\delta B} \wedge [\lambda, B] \right) \end{equation}
Let us show that this is the Poisson bracket $\{ H_{\lambda}, f \}$ where

\begin{equation} H_{\lambda} = \int_{\partial M} \text{tr}(\lambda d_A B). \end{equation}
The Poisson bracket on functions on $F_{\partial M}$ is given by

\begin{equation} \{ f, g \} = \int_{\partial M} \text{tr} \left( \frac{\delta f}{\delta A} \wedge \frac{\delta g}{\delta B} - \frac{\delta g}{\delta A} \wedge \frac{\delta f}{\delta B} \right). \end{equation}

We have
\begin{equation} \frac{\delta H_{\lambda}}{\delta A} = \frac{\delta}{\delta A} \left( \int_{\partial M} \text{tr}(\lambda \, dB + \lambda [A \wedge b]) \right) = [\lambda, B] \end{equation}
and, using integration by parts:

\begin{equation} \frac{\delta H_{\lambda}}{\delta B} = d_A B = dB + [A \wedge B] \end{equation}
This proves the statement.

An important corollary of this fact is that the Hamiltonian action of $G_M$ induces a moment map $\mu : F_{\partial M} \to \mathfrak{g}_{\partial M}^{\ast}$, and it is clear that
\[
C_{\pa M}=\mu^{-1}(0)
\]
This implies that $C_{\pa M}\subset F_{\pa M}$ {\it is a coisotropic submanifold}.

\begin{remark}
Let us show directly that $C_{\partial M} \subset F_{\partial M}$ is a coisotropic subspace of the symplectic space $F_{\partial M}$ when $\mathfrak{g} = \mathbb{R}$. We need to show that
$C_{\partial M}^{\perp}\subset C_{\partial M}$ where $C^\perp$ is the symplectic orthogonal to $C$.

The subspace $C_{\partial M}^{\perp}$ consists of all $(\alpha, \beta) \in \Omega^1(\partial M) \oplus \Omega^{n-2}(\partial M)$ such that

\begin{equation} \int_{\partial M} a \wedge \beta + \int_{\partial M} \alpha \wedge b = 0 \end{equation}

for all $(a, b) \in C_{\partial M} \subset \Omega^1(\partial M) \oplus \Omega^{n-2}(\partial M)$. This condition for all $a$ gives that $\beta = 0$ and requiring this condition for all $b$ gives that $\alpha$ is exact, so we have $C_{\partial M}^{\perp} = \Omega^1_{\text{ex}}(\partial M) \subset C_{\partial M}$ as desired.
\end{remark}

\subsubsection{} The differential $\delta S_M$ of the action functional is the sum of the bulk term
defining the Euler-Lagrange equations and of the boundary term defining
the 1-form $\alpha_{\pa M}$ on the space of boundary fields.
The bulk term vanishes on solutions of the Euler-Lagrange equations, so we have

\begin{equation}
\delta S_M |_{EL_M} = \pi^{\ast}(\alpha_{\partial M}|_{L_M})
\end{equation}
where $\pi : F_M \to F_{\partial M}$ is the restriction to the boundary and $L_M = \pi(EL_M)$. This is analogous to the property of the Hamilton-Jacobi action in classical mechanics.

Because $S_M$ is gauge invariant, it defines the functional on gauge classes of fields and thus, on
gauge classes of solutions to Euler-Lagrange equations.
Passing to gauge classes we now replace the chain of inclusions of gauge invariant subspaces
$L_M \subset C_{\partial M} \subset F_{\partial M}$ with the chain of inclusions of
corresponding gauge classes

\begin{equation} L_M/G_{\partial M} \subset C_{\partial M}/G_{\partial M} \subset F_{\partial M}/G_{\partial M}. \end{equation}

The rightmost space is a Poisson manifold since the action of $G_{\partial M}$ is Hamiltonian. The middle space is the Hamiltonian reduction of $C_{\partial M}$ and is a symplectic leaf in the rightmost space. The leftmost space is still Lagrangian by the standard arguments from symplectic geometry.

\subsubsection{} A natural Lagrangian fibration $p_\pa: \Omega^{n-2}(\pa M)\oplus \Omega^1(\pa M)\to \Omega^1(\pa M)$
corresponds to the Dirichlet boundary conditions when we fix the pull-back of $A$ to the boundary: $a=i^*(A)$.
Such boundary conditions are compatible with the gauge action. Another example of the family of gauge invariant boundary conditions
corresponds to Neumann boundary conditions and is given by the Lagrangian fibration $p_\pa: \Omega^{n-2}(\pa M)\oplus \Omega^1(\pa M)\to \Omega^{n-2}(\pa M)$.

\subsection{Classical Chern-Simons theory}\label{cCS}

\subsubsection{} Spacetimes for classical Chern-Simons field theory are smooth, compact, oriented $3$-manifolds. Let $M$ be such manifold fields $F_M$ on $M$ are connections on the trivial $G$-bundle over $M$ with $G$ being compact, semisimple, connected, simply connected Lie group. We will identify the space of connections with the space of $1$-forms $\Omega^1(M, \mathfrak{g})$. The action functional is

\begin{equation} S(A) = \int_M \text{tr} \left( \frac{1}{2} A \wedge dA + \frac{1}{3} A \wedge A \wedge A \right) \end{equation}
where $A$ is a connection.

The variation is
\begin{equation} \delta S_M(A) = \int_M \text{tr}(F(A) \wedge \delta A) +\frac{1}{2} \int_{\partial M} \text{tr}(A \wedge \delta A) \end{equation}
so the space of solutions $EL_M$ to the Euler-Lagrange equations is the space of flat connections:
\[
 EL_M=\{A| F(A)=0\}
\]
The boundary term defines the $1$-form on boundary fields (connections on the trivial
$G$-bundle over the boundary which we will identify with $\Omega^1(\pa M))$:

\begin{equation} \alpha_{\partial M} = -\frac{1}{2}\int_{\partial M} \text{tr}(A \wedge \delta A) \end{equation}

This 1-form on boundary fields defines the symplectic structure on the space
of boundary fields:
\begin{equation}\label{cs-omega}
\omega_{\pa M}=\delta\alpha_{\pa M}=-\frac{1}{2}\text{tr} \int_{\pa M} \delta A\wedge \delta A
\end{equation}

\subsubsection{} The gauge group $G_M$ is the group of bundle automorphisms of of the trivial principal $G$-bundle over $M$. It can be naturally
be identified with the space of smooth maps $M \to G$ which transform connections as
in (\ref{g-conn}) and we have:

\begin{equation} \label{g-cs-act}S_M(A^g) = S_M(A) + \frac{1}{2}\text{tr}\int_{\pa M} (g^{-1}Ag\wedge g^{-1}dg)-\frac{1}{6}\text{tr} \int_M g^{-1} dg  \wedge g^{-1} dg  \wedge g^{-1} dg . \end{equation}

Assume the integrality of the Maurer-Cartan form on $G$:
\[
\theta=-\frac{1}{6}\text{tr}(dg\, g^{-1} \wedge dg\, g^{-1} \wedge dg\, g^{-1})
\]
i.e. we assume that the normalization of the Killing form is chosen in such a way that
$[\theta]\in H^3(M, \ZZ)$.
Then for a closed manifold $M$ the expression
\[
W_M(g)=-\frac{1}{6}\text{tr} \int_M dg\, g^{-1} \wedge dg\, g^{-1} \wedge dg\, g^{-1}
\]
is an integer and therefore $S_M\mod \ZZ$ is gauge invariant (for details see for example
\cite{Fr}).

\begin{proposition} When the manifold $M$ has a boundary, the functional $W_M(g)\mod \ZZ$
depends only on the restriction of $g$ to $\pa M$.
\end{proposition}

Indeed, $M'$ be another manifold with the boundary $\pa M'$ which differs from
$\pa M$ only by reversing the orientation, so that the result of the gluing $M\cup M'$
along the common boundary is smooth. Then
\[
W_M(g)-W_{M'}(g')=-\frac{1}{6}\text{tr}\int_{M\cup M'} \int_M d\tilde{g} \tilde{g}^{-1} \wedge d\tilde{g} \tilde{g}^{-1} \wedge d\tilde{g} \tilde{g}^{-1} \in \ZZ
\]
Here $\tilde{g}$ is the result of gluing maps $g$ and $g'$ into a map $M\cup M'\to G$.
Therefore, modulo integers, it does not depend on $g$ and $g'$.

For $a$ a connection on the trivial principal $G$-bundle over a 2-dimensional
manifold $\Sigma$ and for $g\in C^\infty(\Sigma, G)$ define
\[
c_\Sigma(a,g)=\exp\left(2\pi i \left(\frac{1}{2}\int_{\pa M} \text{tr}(g^{-1}ag\wedge g^{-1}dg)+W_\Sigma(g)\right)\right)
\]
Here we wrote $W_\Sigma(g)$ because $W_M(g)\mod\ZZ$ depends only on the value of $g$ on $\pa M$.

The transformation property (\ref{g-cs-act}) of the Chern-Simons action
implies that the functional
\[
\exp(2\pi i S_M(A))
\]
transforms as
\[
\exp(2\pi i S_M(A^g))=\exp(2\pi i S_M(A))c_{\pa M}(i^*(A),i^*(g))
\]
where $i^*$ is the restriction to the boundary (pull-back).
For further details on gauge aspects of Chern-Simons theory see
\cite{Fr}\cite{Himpel}.

Now we can define the gauge invariant version of the Chern-Simons action. Consider the trivial circle
bundle $\cL_M=S^1\times F_M$ with the natural projection $\cL_M\to F_M$. Define the action of $G_M$
on $\cL_M$ as
\[
g: (\lambda, A)\mapsto (\lambda c_{\pa M}(i^*(A), i^*(g)), A^g)
\]
The functional $\exp(2\pi i S_M(A))$ is a $G_M$-invariant section of this bundle.
The restriction of $\cL_M$ to the boundary gives the trivial $S^1$-bundle over $F_{\pa M}$
with the $G_{\pa M}$-action
\[
g:(\lambda, A)\mapsto (\lambda c_{\pa M}(A, g), A^g)
\]
The $1$-form $\alpha_{\pa M}$ is a $G_{\pa M}$-invariant connection of $\cL_{\pa M}$.
The curvature of this connection is the $G_{\pa M}$-invariant symplectic form $\omega_{\pa}$.

By definition of $\alpha_{\pa M}$ we have the Hamilton-Jacobi property of the action:

\begin{equation} \delta S_M|_{EL_M} = \pi^{\ast}(\alpha_{\partial M}|_{L_M}) \end{equation}

\subsubsection{} Now, when the gauge symmetry of the Chern-Simons theory is clarified, let us pass to
gauge classes. The action of boundary gauge transformations on $F_{\partial M}$ is Hamiltonian with respect to the symplectic form (\ref{cs-omega}). It is easy to check (and it is well known) that the vector field
on $F_{\pa M}$ generating infinitesimal gauge transformation $A\to A+d_A\lambda$
is Hamiltonian with the generating function

\begin{equation} H_{\lambda}(A) = \int_{\partial M} \text{tr}(F(A) \lambda). \end{equation}

This induces the moment map $\mu : F_{\partial M} \to \mathfrak{g}_{\partial M}^{\ast}$ given by $\mu(A)(\lambda) = H_{\lambda}(A)$.

Let $C_{\partial M}$ be the space of Cauchy data, i.e. boundary values of connections which are flat in a small neighborhood of the boundary. It can be naturally identified with the space of flat $G$-connections on $\partial M$ and thus, $C_{\pa M}=\mu^{-1}(0)$. Hence $C_{\partial M}$ is a coisotropic submanifold of $F_{\partial M}$. We have a chain of inclusions

\begin{equation} L_M = \pi(EL_M) \subset C_{\partial M} \subset F_{\partial M} \end{equation}
where $L_M$ is the space of flat connections on $\partial M$ which extend to flat connections on $M$.
Using the appendix from \cite{CMR} one can easily show that $L_M$ is Lagrangian.

We have following inclusions of the spaces of gauge classes
\begin{equation} L_M/G_{\partial M} \subset C_{\partial M}/G_{\partial M} \subset F_{\partial M}/G_{\partial M} \end{equation}
where the middle term is the Hamiltonian reduction $\mu^{-1}(0)/G_{\partial M} \cong \underline{C}_{\partial M}$, which is symplectic. The left term is Lagrangian, and the right term is Poisson. Note that the middle term is
a finite dimensional symplectic leaf of the infinite dimensional Poisson manifold $F_{\partial M}/G_{\partial M}$.

The middle term $C_{\partial M}/G_{\partial M}$ is the moduli space $\mathcal{M}^G_{\partial M}$ of flat $G$-connections on $\partial M$. It is naturally isomorphic to the representation variety:
\[
\mathcal{\pa M}^G_{\pa M} \cong \text{Hom}(\pi_1(\pa M), G)/G
\]
where $G$ acts on $\text{Hom}(\pi_1(M), G)$ by conjugation. We
will denote the symplectic structure on this space by $\underline{\omega}_{\pa M}$.

Similarly, we have $EL_M/G_M=\mathcal{M}^G_M \cong \text{Hom}(\pi_1(M), G)/G$, which is the moduli space of flat $G$-connections on $M$. Unlike in Yang-Mills case, these spaces are finite-dimensional.

The image of the natural projection $\pi : \mathcal{M}^G_M \to \mathcal{M}^G_{\partial M}$ is the reduction
of $L_M$ which we will denote by  $\underline{L}_M=L_M/G_M$.

Reduction of $\cL_M$ and of $\cL_{\pa M}$ gives line bundles $\underline{\cL}_M=\cL_M/G_M$ and $\underline{\cL}_{\pa M}=\cL_{\pa M}/G_{\pa M}$ over $\mathcal{M}^G_M$ and $\mathcal{M}^G_{\partial M}$ respectively.
The $1$-form $\alpha_{\partial M}$ which is also
a $G_{\pa M}$-invariant connection on $\cL_{\pa M}$ becomes a connection on $\underline{\cL}_{\pa M}$ with the curvature $\underline{\omega}_{\pa M}$.

The Chern-Simons action yields a section $cs$ of the pull-back of the line bundle $\cL_{\pa M}$ over $\mathcal{M}^G_{\partial M}$. Because $\underline{L}_M$
is a Lagrangian submanifold, the symplectic form $\underline{\omega}_{\pa M}$ vanishes on it
and the restriction of the connection
$\underline{\alpha_{\partial M}}$ to $\underline{L}_M$ results in a flat connection over $\cL_{\pa M}|_{\underline{L}_M}$. The section $cs$ is horizontal with respect to the pull-back of the connection
$\underline{\alpha_{\partial M}}$. It can be written as
\begin{equation} (d -\pi^{\ast}(\underline{\alpha_{\partial M}}|_{L_M} ))cs = 0 \end{equation}
This collection of data is the \emph{reduced Hamiltonian structure} of the Chern-Simons theory.

\subsubsection{} There are no natural non-singular Lagrangian fibrations on the space of connections
on the boundary which are compatible with the gauge action. However, for formal
semiclassical quantization we need only such fibration near a smooth point in the space of connections.
We will return to this later, in section \ref{NAqcs}. Now we will describe another structure on
the space of boundary fields for the Chern-Simons theory which is used in geometric quantization \cite{AdPW}.

Instead of looking for a real Lagrangian fibration, let us choose a complex polarization
of $\Omega^1(M,\g)_{\mathbb C}$. Fixing a complex structure on the boundary, gives us the natural
decomposition
\[
\Omega^1(\pa M,\g)_{\mathbb C}=\Omega^{1,0}(\pa M,\g)_{\mathbb C}\oplus \Omega^{0,1}(\pa M,\g)_{\mathbb C}
\]
and we can define boundary fibration as the natural projection to $\Omega^{1,0}(\pa M,\g)_{\mathbb C}$.
Here elements of $\Omega^{1,0}(\pa M,\g)_{\mathbb C}$ are $\g_{\mathbb C}$-valued forms which locally can be written as $a(z,\overline{z})dz$ and elements of $\Omega^{0,1}(\pa M,\g)_{\mathbb C}$ can be written as $b(z,\overline{z})d\overline{z}$. The decomposition above locally works as
follows:
\[
A=\cA+\overline{\cA}
\]
where $\cA=a(z,\overline{z})dz$.

In terms of this decomposition the symplectic form is
\[
\omega=\int_{\pa M} \text{tr} \  \delta \cA\wedge \delta \overline{\cA}
\]
It is clear that subspaces $\cA + \Omega^{0,1}(\pa M)$ are Lagrangian in the complexification of
$\Omega(M, \g)$. Thus, we have Lagrangian fibration $\Omega(M,\g)_\CC\to \Omega^{0,1}(M,\g)_\CC$.
The action of the gauge group preserves the fibers.

However, the form $\alpha_{\pa M}$ does not vanish of these fibers. To make it vanish
we should modify the action as
\[
\tilde{S}_M=S_M+\frac{1}{2}\int_{\pa M} \text{tr}   \ (\cA\wedge \overline{\cA})
\]
After this modification, the boundary term in the variation of the action
gives the form
\[
\tilde{\alpha}_{\pa M}=-\int_{\pa M} \text{tr} \ ( \overline{\cA}\wedge \delta{\cA})
\]
This form vanishes on fibers. It is not gauge invariant as well as the modified action.
The modified action transforms under gauge transformations as
\[
\tilde{S}_M(A^g) = \tilde{S}_M(A) + \frac{1}{2}\text{tr}\int_{\pa M} (g^{-1}\cA g\wedge g^{-1}\overline{\pa}g)+W_M(g)
\]
This gives the following cocycle on the boundary gauge group
\[
\tilde{c}_\Sigma(A,g)=\exp(2\pi i (\frac{1}{2}\int_{\Sigma} \text{tr}(g^{-1}\cA g\wedge g^{-1}\overline{\pa}g)+W_\Sigma(g)))
\]
This modification of the action and this complex polarization of the space of boundary
fields is important for geometric quantization in Chern-Simons theory \cite{AdPW} and
is important for understanding the relation between the Chern-Simons theory and the WZW theory, see for example \cite{EM},\cite{ABM}.
We will not expand this direction here, since we are interested in formal
semiclassical quantization where real polarizations are needed.

\subsection{BF-theory} Space time $M$ is smooth, oriented\footnote{ The orientability assumption can be dropped, see \cite{CR}} and compact and is equipped with a trivial $G$-bundle where $G$ is connected, simple or abelian  compact Lie group.  Fields are

\begin{equation} F_M = \Omega^1(M, \mathfrak{g}) \oplus \Omega^{n-2}(M, \mathfrak{g}) \end{equation}
where $\Omega^1(M, \mathfrak{g})$ describes connections on the trivial $G$-bundle.

The action functional of the BF theory is the topological term of Yang-Mills action:

\begin{equation} S_M(A, B) = \int_M \text{tr}(B \wedge F(A)) \end{equation}

For the variation of $S_M$ we have:

\begin{equation} \delta S_M = \text{tr}\int_M \delta B\wedge F(A) + (-1)^{n-1}\text{tr}\int_M d_AB\wedge \delta A + (-1)^{n-1}\text{tr} \int_{\partial M} B \wedge \delta A \end{equation}
The bulk term gives Euler-Lagrange equations:

\begin{equation} EL_M = \{ (A, B) : F(A) = 0, \ \ d_A B = 0 \}. \end{equation}
The boundary term gives a $1$-form on the space of boundary fields $F_{\partial M} = \Omega^1(\partial M, \mathfrak{g}) \oplus \Omega^{n-2}(\partial M, \mathfrak{g})$:

\begin{equation} \alpha_{\partial M} = \int_{\partial M} \text{tr}(B \wedge \delta A) \end{equation}

The corresponding exact symplectic form is

\begin{equation} \omega_{\partial M} = \delta \alpha_{\partial M} = \int_{\partial M} \text{tr}(\delta B \wedge \delta A). \end{equation}
The space of Cauchy data is
\[
C_{\pa M}=\{(A,B)| F_A=0, \ \ d_AB=0\}
\]
Boundary values of solutions of the Euler-Lagrange equations on $M$ define the submanifold
$L_M=\pi(EL_M)\subset F_{\pa M}$. This submanifold is Lagrangian.
Thus we have the embedding:
\[
L_M\subset C_{\pa M}\subset F_{\pa M}
\]
where $F_{\pa M}$ is exact symplectic, $C_{\pa M}$ is co-isotropic, and $L_M$ is Lagrangian.

\subsubsection{} The space of bundle automorphisms $G_M$ is the space of smooth maps $M \to G$. They act on $A \in \Omega^1(M, \mathfrak{g})$ by $A \mapsto g^{-1} Ag + g^{-1} dg$ and on $B \in \Omega^{n-2}(M, \mathfrak{g})$ by $B \mapsto g^{-1} Bg$. As in Yang-Mills theory the action is invariant with respect to these transformations.

In addition, it is almost invariant with respect to transformations $A \mapsto A, B \mapsto B + d_A \beta$ where $\beta \in \Omega^{n-3}(M, \mathfrak{g})$:

\begin{equation} S_M(A, B + d_A \beta) = S_M(A, B) + \int_M \text{tr}(d_A \beta \wedge F(A)) \end{equation}
After integration by parts in the second term we write it as

\begin{equation} \int_M \text{tr}(\beta \wedge d_A F(A)) + \int_{\partial M} \text{tr}(\beta \wedge F(A)). \end{equation}
The bulk term here vanishes because of the Bianchi identity and the only additional contribution is a boundary term, thus:
\[
S_M(A, B + d_A \beta) = S_M(A, B) + \text{tr} \int_{\partial M} (\beta \wedge F(A))
\]
The additional gauge symmetry $B \mapsto B + d_A \beta$ gives us a larger gauge group

\begin{equation} G_M^{BF} = G_M \times \Omega_M^{n-3} \end{equation}
Its restriction to the boundary gives the boundary gauge group

\begin{equation} G_{\partial M}^{BF} = G_{\partial M} \times \Omega_{\partial M}^{n-3}. \end{equation}

The action is invariant up to a boundary term. This means that the $1$-form $\alpha_{\pa M}$ is not gauge invariant.
Indeed, it is invariant with respect to $G_M$-transformations, but when $(A,B)\mapsto (A, B+d_A\beta)$
the forms $\alpha_{\pa M}$ transforms as
\[
\alpha_{\pa M}\mapsto \alpha_{\pa M}+\int_{\pa M} \text{tr} \ d_A\beta\wedge \delta A
\]
However, it is clear that the symplectic form $\omega_{\pa M}=\delta\alpha_{\pa M}$ is gauge invariant. Moreover, we have the following.
\begin{theorem} The action of $G_{\partial M}^{BF}$ is Hamiltonian. \end{theorem}

Indeed, if $\alpha \in \Omega^0(\partial M, \mathfrak{g})$  is an element of the Lie algebra of boundary
 gauge transformations and $\beta \in \Omega^{n-3}(\partial M, \mathfrak{g}$, then we can take

\begin{equation} H_{\alpha}(A, B) = \int_{\partial M} \text{tr}(B \wedge d_A \alpha) \end{equation}
\begin{equation} H_{\beta}(A, B) = \int_{\partial M} \text{tr}(A \wedge d_A \beta). \end{equation}
as Hamiltonians generating the action of corresponding infinite dimensional Lie algebra.

This defines a moment map $\mu : F_{\partial M} \to \Omega^0(\partial M, \mathfrak{g}) \oplus \Omega^{n-3}(\partial M, \mathfrak{g})$. It is clear that Cauchy submanifold is also $C_{\partial M} = \mu^{-1}(0)$.
This proves that it is a co-isotropic submanifold.

Note also, that the restriction of $\alpha_{\pa M}$ to $C_{\pa M}$ is $G^{BF}_{\pa M}$-invariant.
Indeed $\text{tr} \int_{\pa M} d_A\beta\wedge \delta A=- \text{tr} \int_{\pa M} \beta\wedge d_A\delta A$,
and this expression vanishes when the form is pulled-back to the space of flat connections
where $d_A\delta A=0$.
Therefore the Hamiltonian reduction of $F_{\pa M}$ which is $\underline{F}_{\pa M}=C_{\pa M}/ G^{BF}_{\pa M}$
is an exact symplectic manifold.

It is easy to see that the reduced space of fields on the boundary $\underline{F}_{\pa M}$ can be
naturally identified, as a symplectic manifold, with $T^*\mathcal{M}_{\pa M}^G$, the cotangent
bundle to the moduli space of flat connections $\mathcal{M}_{\pa M}^G=\text{Hom}(\pi_1(\pa M),G)/G$.
The canonical $1$-form on this cotangent bundle corresponds to the form $\alpha_{\pa M}$ restricted
to $C_{\pa M}$. The Lagrangian subspace $L_M\subset F_{\pa M}$ is gauge invariant. It defines the
Lagrangian submanifold
\[
L_M/G^{BF}_{\pa M}\subset  T^*\mathcal{M}_{\pa M}^G
\]

The restriction of the action functional to $EL_M$ is gauge invariant and defines
the the function $\underline{S}_M$ on $EL_M/G^{BF}_{\pa M}$. The formula
for the variation of the action gives the analog of the Hamilton-Jacobi formula.

\begin{equation} d \underline{S}_M = \pi^{\ast}(\theta|_{L_M}). \end{equation}
where $\theta$ is the canonical $1$-form on the cotangent bundle $T^*\mathcal{M}_{\pa M}^G$
restricted to $L_M/G^{BF}_{\pa M}$.

\subsubsection{} One of the natural choices of boundary conditions is the Dirichlet boundary conditions.
This is the Lagrangian fibration $\Omega^1(M,\g)\oplus \Omega^{n-2}(M,\g)\to  \Omega^1(M,\g)$.
This fibration is gauge invariant. After the reduction it gives the standard
Lagrangian fibration $T^*\mathcal{M}_{\pa M}^G\to \mathcal{M}_{\pa M}^G$.

\section{Semiclassical quantization of first order field theories}

\subsection{The framework of local quantum field theory}

We will follow the framework of local quantum field theory which was outlined
by  Atiyah and Segal for topological and conformal field theories.
In a nut-shell it is a functor from a category of cobordisms to the
category of vector spaces (or, more generally, to some
category).

All known local quantum field theories can be formulated in
this way at some very basic level. It does not mean that this is a final destination of our understanding of quantum dynamics at the microscopical scale. But at the moment this general setting includes the standard model, which agrees with most of the experimental data in high energy physics. In this sense this is the accepted framework at the moment, just as at different points of history, classical mechanics, classical electro-magnetism, and quantum mechanics were playing such a role.

A quantum field theory in a given space time category gives
a functor from this category to the category of vector
spaces (or to another ``known'' category).
It can also be regarded as an assignment of a
vector space to the boundary of the space time manifold and a vector in this vector space to
the manifold:
\[
N\mapsto H(N), \ \ M\mapsto Z_M\in H(\pa M).
\]
The identification of such assignments with linear maps
is natural assuming that the vector space assigned to the
boundary is the tensor product of vector spaces
assigned to connected components of the boundary and
that changing the orientation replaces the corresponding vector space by its dual.

The vector space assigned to the boundary is the space of boundary
states. It may depend on the extra structure at
the boundary. In this case it is a vector bundle over the space of admissible geometric data and $Z_M$ is a
section of this vector bundle. The vector $Z_M$ is called the  partition
function or the amplitude.

These data should satisfy natural axioms, which
can by summarized as follows:
\begin{enumerate}

\item The locality properties of boundary states:
\[
H(\O)=\CC\,,\quad H(N_1\sqcup N_2)=H(N_1)\otimes H(N_2),
\]
\item The locality property of the partition function
\[
Z_{M_1\sqcup M_2}=Z_{M_1}\otimes Z_{M_2}\in H(\pa M_1)\otimes
H(\pa M_2).
\]

\item For each space $N$ (an object of the space time category)
there is a non-degenerate pairing
\[
\left< .,.\right> _N: H(\overline{N})\otimes H(N)\to \CC
\]
such that $\left< .,.\right> _{N_1\sqcup N_2}=\left< .,.\right> _{N_1}\otimes \left< .,.\right> _{N_2}$.

\item The orientation reversing automorphism $\sigma: N\to \overline{N}$
lifts to a $\CC$-antilinear mapping $\widehat{\sigma}_N: H(N)\to H(\overline{N})$
which agrees with locality of $N$ and  $\widehat{\sigma}_{\overline{N}}\widehat{\sigma}_N=id_N$.
Together with the pairing $\left< .,.\right> _N$ the orientation reversing mapping
induces the Hilbert space structure on $H(N)$.

\item An orientation preserving isomorphism\footnote{By an isomorphism here we mean a mapping preserving
the corresponding geometric structure} $f:N_1\to N_2$ induces a linear isomorphism
\[
T_f: H(N_1)\to H(N_2).
\]
which is compatible with the pairing and $T_{f\sqcup g}=T_f\otimes T_g$, \ \ $T_{f\circ g}=T_fT_g$ (possibly with a cocycle).

\item The gluing axiom. This pairing should agree with
partition functions in the following sense. Let $\pa M=N\sqcup\overline{N}\sqcup{N'}$, then
\begin{equation}\label{gluing-prop}
(\langle.,.\rangle\otimes id)Z_{M}=Z_{M_N}\in H(N')
\end{equation}
where $M_N$ is the result of gluing of $M$ along $N$.
The operation is known as the gluing axiom. For more details see
\cite{BK}.

\item The quantum field theory is (projectively) invariant with respect to
transformations of the space time (diffeomorphisms, gauge transformations etc.)
if for such transformation $f: M_1\to M_2$,
\[
T_{f_\pa} Z_{M_1}=c_{M_1}(f) Z_{M_2}
\]
Here $c_M(f)$ is a co-cycle $c_M(fg)=c_{gM}(f)c_M(g)$. When the theory is invariant, not only projectively invariant, $c_M(f)=1$.

\end{enumerate}

\begin{remark}
The gluing axiom in particular implies the
functoriality of $Z$:
\[
Z_{M_1\circ M_2}=Z_{M_1}\circ Z_{M_2}\,.
\]
Here $M_1\circ M_2$ is the composition of cobordisms
in the category of space time manifolds. In case of cylinders
this is the semigroup property of propagators in the operator formulation of QFT.
\end{remark}

\begin{remark}
This framework is very natural in models of statistical mechanics on cell complexes
with open boundary conditions, also known as lattice models.
\end{remark}

\begin{remark}
The main physical concept behind this framework is the locality
of the interaction. Indeed, we can cut our space time manifold in
small pieces and the resulting partition function $Z_M$ in such framework is expected to be the composition of partition
functions of small pieces. Thus, the theory is determined by
its structure on `small' space time manifolds, or at `short
distances'. This is the concept of {\it locality}. To fully implement this concept one
should consider the field theory on manifolds with corners where we can glue
along parts of the boundary.
\end{remark}

\subsection{Path integral and its finite dimensional model}

\subsubsection{Quantum field theory via path integrals}
Given a first order classical field theory with boundary conditions given by Lagrangian fibrations,
one can try to construct a quantum field theory by the path integral quantization. In this
framework the space of boundary states $H(\partial M)$ is taken as  the space of functionals on the base $B_{\partial M}$ of the Lagrangian fibration on boundary fields $F_{\pa M}$. The vector $Z_M$ is the Feynman integral
over the fields on the bulk with given boundary conditions

\begin{equation} Z_M(b) = \int_{f \in \pi^{-1} p^{-1}(b)} e^{ \frac{i}{h} S_M(f) } Df \end{equation}
where $Df$ is some fantasy measure, $\pi: F_M\to F_{\pa M}$ is the restriction map and
$p: F_{\pa M}\to B_{\pa M}$ is the boundary fibration.

All of the above is difficult to define when the space of fields is infinite
dimensional. To clarify the functorial structure of this construction and to
define the formal semiclassical path integral let us start with a model case when the space of fields
is finite dimensional, when the integrals are
defined and absolutely convergent. A ``lattice approximation" of a continuous
theory is a good example of such a finite dimensional model.

\subsubsection{Finite dimensional classical model} \label{fd-class}

A finite dimensional model of a first order classical field theory on a space time
manifold with boundary consists of the
following data. Three finite dimensional manifolds $F, F_\pa, B_\pa$ (more generally, topological
spaces) should be complemented by the following.

\begin{itemize}

\item The space $F_\pa$ is an exact symplectic manifold with symplectic form
$\omega_\pa=d\alpha_\pa$.

\item The function $S$ on $F$, such that the submanifold $EL\subset F$, on which the form $dS-\pi^*(\alpha_\pa)$ vanishes, projects to a Lagrangian submanifold in $F_\pa$.

\item Lagrangian fibration on $p_\pa: F_\pa\to B_\pa$ which is $\alpha_\pa$ exact (i.e. $\alpha_\pa$ vanishes on fibers). We also assume that a generic fiber is transversal to $L=\pi(EL)\subset F_\pa$.

\end{itemize}

We will say that this is a finite dimensional model of a {\it non-degenerate theory} if
$S$ has finitely many simple critical points
on each generic fiber $\pi^{-1}p_\pa^{-1}(b)$.

The model is gauge invariant with the bulk gauge group $G$ and the boundary gauge group $G_\pa$
if the following holds.
\begin{itemize}

\item The group $G$ acts on $F$, and $G_\pa$ acts on $F_\pa$.
\item There is a group homomorphism $\tilde{\pi}: G\to G_\pa$ such that the
restriction map is a map of group manifolds, i.e. $\pi(gx)=\tilde{\pi}(g)\pi(x)$.
\item The function $S$ is invariant under the $G$-action up to boundary terms:
\[
S(gx)=S(x)+c_\pa(\pi(x), \tilde{\pi}(g))
\]
where $c_\pa(x,g)$ is a cocycle for $G_\pa$ acting on $F_\pa$:
\[
c_\pa(x,gh)=c_\pa(hx,g)+c_\pa(x,h)
\]
\item The action of $G_\pa$ is compatible with the fibration $p_\pa$, i.e.
it maps fibers to fibers. In particular this means that there is a group
homomorphism $\tilde{p}_\pa: G_\pa \to \Gamma_\pa$. In addition we require that
the cocycle $c(g,x)$ is constant on fibers of $p_\pa$, i.e. $c(x,g)=c(p_\pa(x),\tilde{p_\pa}(g))$.

\end{itemize}

We will say that the theory with gauge invariance is non-degenerate if critical points of
$S$ form finitely many $G$-orbits and if the corresponding points on $F(b)/G$ are simple (i.e. isolated)
on each generic fiber $F(b)$ of $p_\pa \pi$.

\subsubsection{Finite dimensional quantum model}\label{nd-glue}

To define quantum theory assume that $F$ and $B_\pa$ are defined together with measures $dx$ and
$db$ respectively. Assume also that there is a measure $\frac{dx}{db}$
on each fiber $F(b)=\pi^{-1}p_\pa^{-1}(b)$ such that $dx=\frac{dx}{db}db$.

Define the vector space $H_\pa$ together with the Hilbert space
structure on it as follows:
\[
H_\pa=L^2(B_\pa)
\]
When the function $S$ is only projectively invariant with respect to the gauge group,
the space of boundary states is the space of $L^2$-sections of the corresponding line bundle.

\begin{remark} It is better to consider the space of half-forms on $B_\pa$ which are
square integrable but we will not do it here. For details see for example \cite{BW}.
\end{remark}

The partition function $Z_F$ is defined as an element of $H_\pa$ given by the integral over the fiber $F(b)$:
\begin{equation}\label{Z}
Z_F(b)=\int_{F(b)} \exp(\frac{i}{h}S(x)) \frac{dx}{db}
\end{equation}

When there is a gauge group the partition function transforms as
\[
Z_F(\gamma b)=Z_F(b) \exp(\frac{i}{h} c_\pa (b, \gamma))
\]

In such a finite dimensional model the gluing property follows from Fubini's theorem
allowing to change the order of integration.
Suppose we have two spaces $F_1$ and $F_2$ fibered over $B_\pa$ and two
functions $S_1$ and $S_2$ defined on $F_1$ and $F_2$ respectively
such that integrals $Z_{F_1}(b)$ and $Z_{F_2}(b)$ converge absolutely for generic $b$.
For example, we can assume that all spaces $F$, $F_\pa$ and $B_\pa$ are compact.
Then changing the order of integration we have
\begin{equation}\label{fd-gl-nd}
\int_{B_\pa} Z_{F_1}(b)Z_{F_2}(b)db=Z_{F_1\times_{B_\pa} F_2}
\end{equation}
where
\[
Z_{F_1\times_{B_\pa} F_2}=\int_{F_1\times_{B_\pa} F_2} \exp(\frac{i}{h}(S_1(x_1)+S_2(x_2)))\frac{dx_1}{db}\frac{dx_2}{db}db
\]
Here $F_1\times_{B_\pa}F_2=\{(x,x')\in F_1\times F_2| \pi_1(x)=\pi_2(x')\}$ is the
fiber product of $F_1$ and $F_2$ over $B_\pa$. The measure $\frac{dx}{db}\frac{dx'}{db}db$
is induced by measures on $F_1(b)$, $F_2(b)$ and on $B_\pa$.

\begin{remark}The quantization is not functorial. We need to make a choice of measure of
integration.
\end{remark}
\begin{remark} We will not discuss here quantum statistical mechanics
where instead of oscillatory integrals we have integrals of probabilistic type
representing Boltzmann measure.
Wiener integral is among the examples of such integrals.
\end{remark}
\begin{remark} When the gauge group is non-trivial, the important subgroup in
the total gauge group is the bulk gauge group, i.e. the symmetry of
the integrand in the formula for $Z_F(b)$. If $\Gamma_\pa$ is the gauge
group acting on the base of the boundary Lagrangian fibration, then the bulk gauge
group $G^B$ is the kernel in the exact sequence of groups $1\to G^B\to G\to \Gamma_\pa\to 1$.
\end{remark}

An example of such construction is the discrete time quantum mechanics
which is described in Appendix \ref{d-time-qm}.

\subsubsection{The semiclassical limit, non-degenerate case}

The asymptotical expansion of the integral (\ref{Z}) can be computed
by the method of stationary phase (see for example \cite{qft-ias},\cite{R} and references
therein).

Here we assume that the function $S$ has
finitely many simple critical points on the fiber $F(b)$ for each
generic $b \in B_\pa$. Denote the set of such critical points by $C(b)$.
Using the stationary phase approximation
we obtain the following expression for the asymptotical expansion of the partition function
as $h\to 0$:

\begin{equation}\label{sem-cl-z-nd}
Z(b)\simeq \sum_{c\in C(b)} Z_c
\end{equation}
where $Z_c$ is the contribution to the asymptotical expansion from the
critical point $c$. To describe $Z_c$ let us choose local coordinates $x^i$
near $c$, then
\begin{equation}\label{Z-scl}
Z_c=(2\pi h)^{\frac{N}{2}}
\frac{1}{\sqrt{|\det(B_c)|}}e^{\frac{iS(c)}{h}+\frac{i\pi}{4}\sign(B_c)}
(v(c)+\sum_{\Ga}\frac{(ih)^{-\chi(\Ga)}F_c(\Ga)}{|\Aut (\Ga)|})
\end{equation}
Here $N=\dim F$ and  $(B_c)_{ij}=\frac{\pa^2 S(c)}{\pa x^i\pa x^j}$, $v(x)$ is the volume density in local coordinates $\{x^i\}_{i=1}^N$ on $F(b)$ , $\frac{dx}{db}=v(x)dx^1\dots dx^N$, $\chi(\Ga)$ is the Euler characteristic
of the graph $\Ga$, $|\textrm{Aut}(\Ga)|$ is the number of automorphisms of the graph and the summation is taken over
finite graphs where each vertex has valency at least $3$. The weight of a graph $F_c(\Ga)$ is
given by the ``state sum" which is described in the appendix \ref{Fd}. Note that this formula is invariant
with respect to change of local coordinates, as it follows from the definition.
This is particularly clear at the level of determinants.
Indeed, let $J$ be the Jacobian of the coordinate change $x^i\mapsto f^i(x)$. Then $v\mapsto v |\det(J)|$
and $|\det(B_c)|\mapsto |\det(B_c)||\det(J)|$ and the Jacobians cancel. For higher level contributions, see \cite{JF}.

\subsubsection{Gluing formal semiclassical partition functions in the non-degenerate case}
The image $L=\pi(EL)$, according to our assumptions is transversal to generic fibers of $p_\pa: F_\pa \to B_\pa$.
By varying the classical background $c$ we can span the  subspace $T_{\pi(c)}L\subset T_{\pi(c)}F_\pa$
which is, according to the assumption of transversality, isomorphic to $T_{\pi(c)}B_\pa$.

We will call the partition function $Z_c$ the {\it formal semiclassical partition function} on
the classical background  $c$.\footnote{The asymptotical formula (\ref{Z-scl})
should be thought as a ``pointwise version" of the following, more natural asymptotical expansion.
Let us say that the family of vectors $\psi=\exp(\frac{i}{h}f)\phi\in H_\pa$ parameterized by $h$
is a semiclassical family of states if $f$ is a smooth function
and $\phi\in H_\pa$. It is easy to see that the semiclassical expansion of
\[
(Z, \psi)=\int_{B_\pa} Z_{B_\pa}(b) \psi(b)db
\]
has a form similar to (\ref{Z-scl}) with some extra vertices corresponding to
derivatives of the state $\psi$.}
We will also say that it is given by the formal integral of $\exp(\frac{iS}{h})$
over the formal neighborhood of $c$:
\[
Z_c=\int^{formal}_{T_cF} \exp(\frac{iS}{h}) \frac{dx}{db}
\]

Passing to the limit $h\to 0$ in (\ref{fd-gl-nd}) we obtain the gluing formula for formal semiclassical
partition functions (under the assumption of non-degenerate critical points):
\[
\int^{formal}_{T_{b_0}B_\pa} Z_{c_1(b)} Z_{c_2(b)}db=Z_c
\]
Here $c$ is a simple critical point of $S$ on $F_1\times_{B_\pa} F_2$, $b_0=p_\pa\pi_1\pi(c)=p_\pa \pi_2\pi'(c)$
where $\pi:  F_1\times_{B_\pa} F_2\to F_1$ and $\pi': F_1\times_{B_\pa} F_2\to F_2$ are natural projections
and $c_1(b)$ and $c_2(b)$ are critical points of $S_1$ and $S_2$ on fibers $F_1(b)$ and $F_2(b)$ respectively
which are formal deformations of $c_1(b_0)=\pi_1(c)$ and of $c_2(b_0)=\pi_2(c)$.
In \cite{JF} this formula was used to prove that formal semiclassical propagator
satisfies the composition property.

\subsection{Gauge fixing}\label{fdg}

\subsubsection{Gauge fixing in the integral} Here we will outline a version of the Faddeev-Popov trick for
gauge fixing in the finite dimensional model in the presence of boundary.
We assume that the action function $S$, the choice of boundary conditions, and
group action on $F$ satisfy all properties described in section \ref{fd-class}.

The goal here is to compute the asymptotics of the partition function
\[
Z_F(b)=\int_{F(b)} e^{\frac{i}{h}S(x)} \frac{dx}{db}
\]
when $h\to 0$. Here, as in the previous section $F(b)=\pi^{-1}p_\pa^{-1}(b)$
but now a Lie group $G$ acts on $F$ and the function $S$ and the integration measure $dx$
are $G$-invariant. As in section \ref{fd-class} we assume that there is an exact
sequence $1\to G^B \to G\to \Gamma_\pa\to 1$, where $\Gamma_\pa$ acts on $B_\pa$
such that $db$ is $\Gamma_\pa$-invariant and the subgroup $G^B$ acts fiberwise
such that the measure $\frac{dx}{db}$ is $G^B$-invariant.

Assume that the function $S$ has finitely many isolated $G^B$-orbits of
critical points on $F(b)$ and that the measure of integration is supported on a
neighborhood of these points.

\begin{figure}[htb]
\includegraphics[height=6cm,width=6cm]{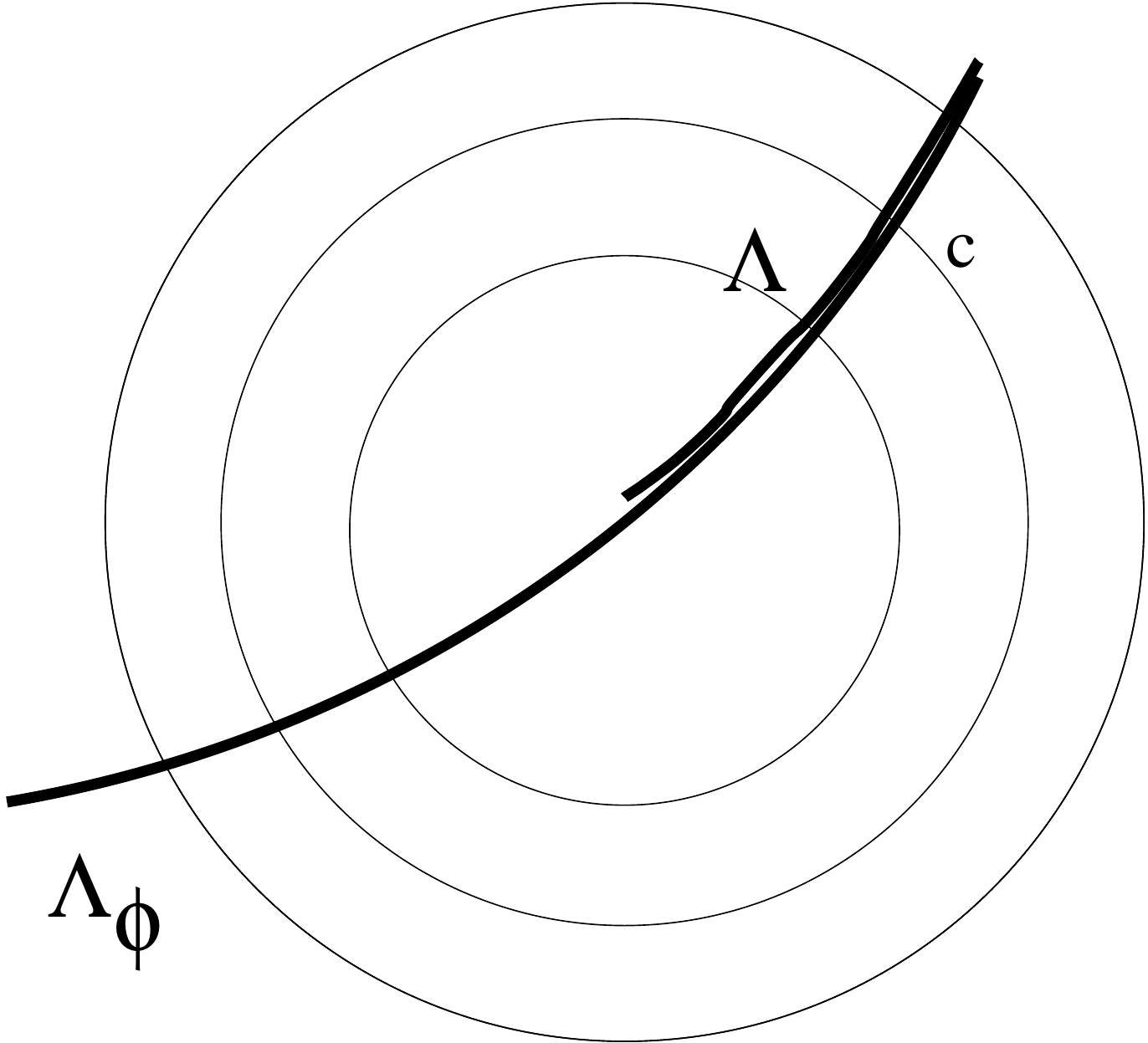}
\caption{An illustration of a choice of $\Lambda$ and $\Lambda_\varphi$.}
\label{gauges}
\end{figure}

Let $\Lambda$ be a submanifold (possibly immersed) in $F$ such that
it intersects the support of $v$ (the density of the measure in local coordinates,
$\frac{dx}{db}=v(x)dx^1\dots dx^N$) over a submanifold which is
also the intersection of $\text{supp}(v)$ and the surface $\Lambda_\varphi=\{x|\varphi_a(x)=0\}$.
Assume that $\Lambda_\varphi$ intersect each orbit in the support of $v$ exactly once.
Note that we do not assume that either $\Lambda$ or $\Lambda_\varphi$ are sections
of the $G^B$-action (i.e. of the projection $F(b)\to F(b)/G^B$).

Let $c$ be a critical point of $S$ on $F(b)$ which also belongs to $\Lambda_\varphi$.
Denote by $U_c^{(b)}$ the connected component of the support of $v$ which contains this critical point.
Following standard steps in the Faddeev-Popov constructions we obtain the
expression for each term of this sum in terms the integral over $\Lambda_\varphi$:
\begin{equation}\label{fp}
\int_{U_c^{(b)}} e^{\frac{i}{h}S(x)} \frac{dx}{db}=|G^B|\int_{U_c^{(b)}}e^{\frac{i}{h}S(x)}  \det(L_\varphi(x))\delta(\varphi(x))\frac{dx}{db}
\end{equation}
We have a natural
isomorphism 
$U_c^{(b)}\simeq (U_c^{(b)}\cap \Lambda_\varphi)\times G^B$.
To describe $L_\varphi(x)$ choose
local coordinates $x^i$ on $F_c$ and a basis $e_a$ in the Lie algebra $\mathfrak{g}^B$
of the Lie group $G^B$. The action of $e_a$ on $F$ is given by the vector field $\sum_il_a^i(x)\pa_i$.
Matrix elements of $L_\varphi(x)$ are $\sum_i l_a^i(x)\pa_i \varphi_b(x)$.

It is convenient to write (\ref{fp}) as a Grassmann integral:
\begin{equation}\label{fp-action}
|G^B|\int_{\cF_c(b)}\exp\left(\frac{i}{h}(S(x)+\sum_a \lambda^a \varphi_a(x))+\sum_a \overline{c}^aL_\varphi(x)_a^bc_b\right)  \frac{dx}{db}d\lambda d\overline{c}dc
\end{equation}
where $\cF_c(b)=
U_c^{(b)}\oplus \g^B_{odd}\oplus (\g^B_{odd})^*\oplus (\g^B_{even})^*$ and $\overline{c}$ and $c$ are
odd variables. See for example \cite{qft-ias} for details on Grassman integration.
To be pedantic, (\ref{fp-action}) also contains a normalization factor $(2\pi h)^{\dim G^B}$.
The asymptotical expansion of (\ref{fp}) as $h\to 0$ can be written as a formal integral over the formal neighborhood
of $c$ in the supermanifold $\cF_c(b)$. The functions $S(x), \varphi_a(x), L_\varphi(x)_a^b$ should
be understood as the Taylor expansion about $c$ in $\sqrt{h}$, just as in the previous section.
The result is the asymptotical expression given by Feynman integrals where two types of edges correspond to
even and odd Gaussian terms in the integral :
\begin{multline}\label{as-FP}
Z_c=\int_{T_cF(b)}^{formal} e^{\frac{i}{h}S(x)} \frac{dx}{db}= |G^B|h^{\frac{N-n}{2}}(2\pi)^{\frac{N+n}{2}}\\
\frac{1}{\sqrt{|\det(B(c))|}}\det(-iL_\varphi(c))
\exp\left(\frac{i}{h}S(c)+\frac{i\pi}{4}\sign(B(c))\right)\\
\left(v(c)+ \sum_{\Ga\neq \O}\frac{(ih)^{-\chi(\Ga)}(-1)^{c(D(\Gamma))}F_c(D(\Ga))}{|\Aut (\Ga)|}\right),
\end{multline}
Here $N=\dim F$ and $n=\dim \g^B$, $D(\Ga)$ is the planar projection of $\Ga$,
a Feynman diagram.
Feynman diagrams in this formula have bosonic edges and
fermionic oriented edges, $c(D(\Ga))$ is the number of crossings of
fermionic edges\footnote{The sign rule is equivalent to the usual $(-1)^{\#\mbox{fermionic loops}}$ which is
used in physics literature.}\footnote{In (\ref{as-FP}) we identify $T_cF(b)$ with an infinitesimally small (formal)
neighborhood of $c$.}. The structure of Feynman diagrams is the same as in
(\ref{Z-scl}). The propagators corresponding to Bose and Fermi
edges are shown in Fig. \ref{FP-diag-edges}. The weights of vertices  are
shown\footnote{Each fermionic propagator
contributes to the weight of the diagram an extra factor $h^{-1}$.
Each vertex with two adjacent fermionic (dashed) edges contributes
the factor of $h$. Because fermionic lines form loops, these factors
cancel each other.}  on Fig. \ref{FP-diag-vert}.

\begin{figure}[htb]
\includegraphics[height=3cm,width=7cm]{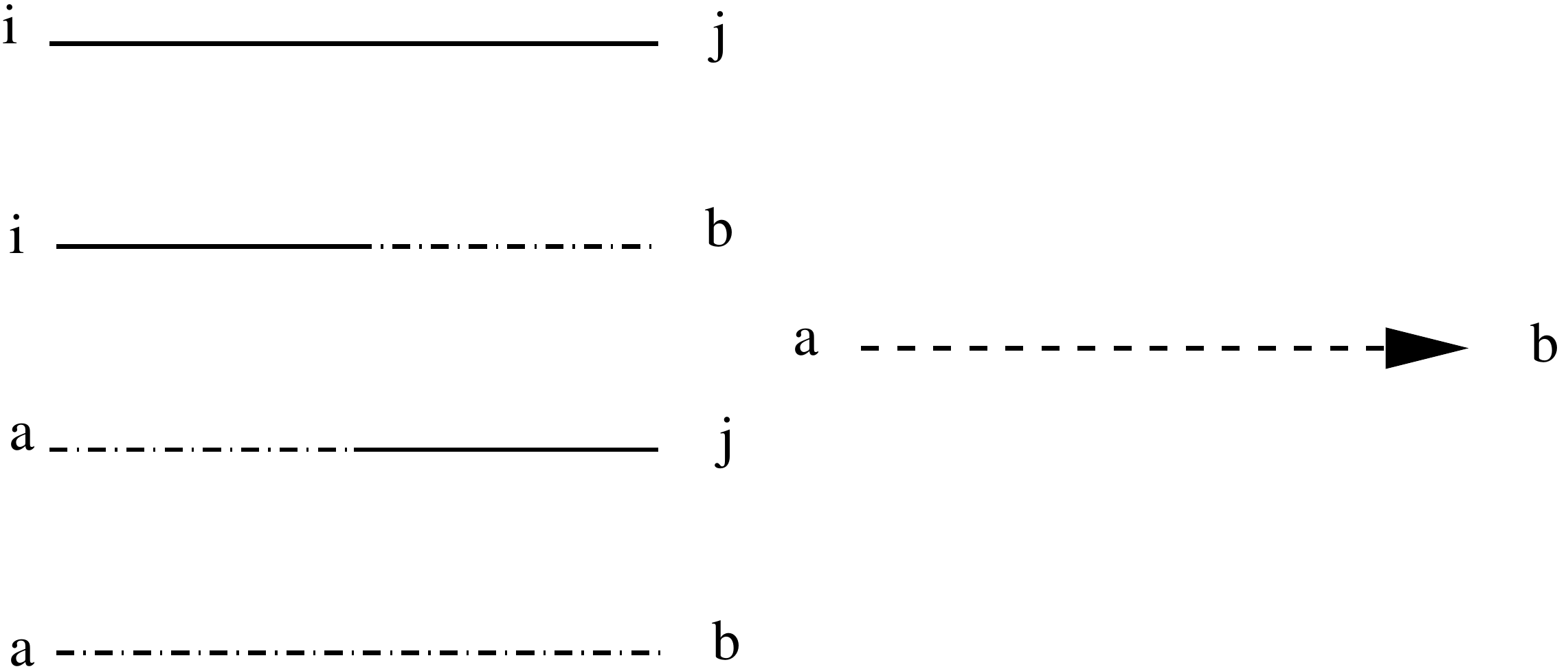}
\caption{Fermionic (left) and bosonic (right) edges for Feynman diagrams in (\ref{as-FP})
with states at their endpoints}
\label{FP-diag-edges}
\end{figure}

\begin{figure}[htb]
\includegraphics[height=3cm,width=8cm]{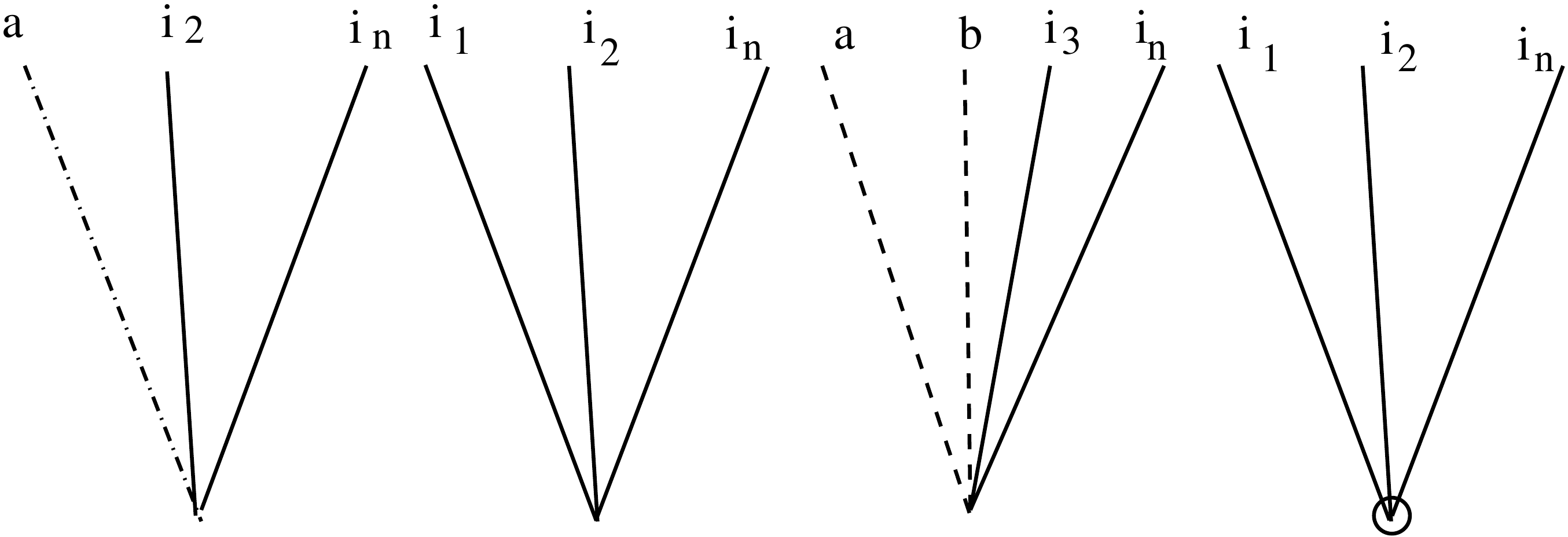}
\caption{Vertices for Feynman diagrams in (\ref{as-FP}) with states
on their stars}
\label{FP-diag-vert}
\end{figure}

The weight of the Fermionic edge on Fig. \ref{FP-diag-edges} is $((-iL_\varphi(c))^{-1})_{ab}$. Weights of the bosonic edges from
Fig. \ref{FP-diag-edges} correspond to matrix element of $B(c)^{-1}$
where
\[
B(c)=\Bigl(
\begin{array}{cc}
\frac{\pa^2 S(c)}{\pa x^i \pa x^j} &  \frac{\pa \varphi_a}{\pa x^i} \\ \frac{\pa \varphi_b}{\pa x^j}  & 0
\end{array}
\Bigl)
\]
The weights of vertices with states on their stars from Fig. \ref{FP-diag-vert}
are (from left to right):
\[
\frac{\pa^{n-1} \varphi_a(c)}{\pa x^{i_1}\dots \pa x^{i_{n-1}}}, \ \
\frac{\pa^{n} S(c)}{\pa x^{i_1}\dots \pa x^{i_n}}, \ \
i\frac{\pa^{n} L^a_b(c)}{\pa x^{i_1}\dots \pa x^{i_n}}, \ \ \frac{\pa^{n} v(c)}{\pa x^{i_1}\dots \pa x^{i_n}}
\]
The last vertex should appear exactly once in each diagram.

This formula, by definition, does not depend on the choice of local coordinates.
It is easy to see this explicitly at the level of determinants. Indeed, when we
change local coordinates
\[
B(c)\mapsto \Bigl(
\begin{array}{cc}
J & 0 \\ 0  & 1
\end{array}
\Bigl)B(c)\Bigl(
\begin{array}{cc}
J & 0 \\ 0  & 1
\end{array}
\Bigl), \ \ v\mapsto |\det(J)|v
\]
where $J$ is the Jacbian of the coordinate transformation. It is clear that the ratio $v/|\det(B(c))|$
is invariant with respect to such transformations.

Note that because we defined the formal integral (\ref{as-FP}) as the contribution to the asymptotical
expansion of the integral (\ref{fp}) from the critical orbit of $S$ passing through $c$,
the coefficients in (\ref{as-FP}) do not depend on the choice of gauge constraint $\varphi$ and
\[
Z_c=Z_{[c]}
\]
where $[c]$ is the orbit of $G^B$ passing through $c$. It is also independent on the
choice of gauge condition $\varphi$. It is easy to see this at the level of determinants.
If we change the gauge condition $\varphi(x)$ to $f(\varphi)$, matrices $B(c)$ and $L(c)$
change as
\[
B(c)\mapsto \Bigl(
\begin{array}{cc}
1 & 0 \\ 0  & \frac{\pa f}{\pa x}(c)
\end{array}
\Bigl)B(c)\Bigl(
\begin{array}{cc}
1 & 0 \\ 0  & \frac{\pa f}{\pa x}(c)
\end{array}
\Bigl), \ \  L(c)\mapsto L(c)\frac{\pa f}{\pa x}(c)
\]
Which implies that the ratio $\det(L(c))/|\det(B(c))|^{1/2}$ remain invariant.

\subsubsection{Gluing formal integrals for gauge theories}
Assume that as in section \ref{nd-glue} we have two spaces $F_1$ and $F_2$ fibered over $B_\pa$ and two
functions $S_1$ and $S_2$ defined on $F_1$ and $F_2$ respectively
such that the integrals $Z_{F_1}(b)$ and $Z_{F_2}(b)$ converge absolutely for generic $b$.
For example, we can assume that spaces $F_1$, $F_2$ and $B_\pa$ are compact.
Denote by $F$ the fiber product $F_1\times_{B_\pa} F_2$ and set $N_i=\dim F_i$, $N_\pa=\dim B_\pa$. Let Lie groups $G_1$, $G_2$ and $\Gamma_\pa$ act as $G_i: F_i\to F_i$ and $\Gamma_\pa : B_\pa\to B_\pa$ and assume that functions $S_i$ are $G_i$-invariant and $\Gamma_\pa$ appears in exact sequences:
\[
0\to G^B_1\to G_1\to \Gamma_\pa \to 0, \ \ 0\to G^B_2\to G_2\to \Gamma_\pa \to 0
\]
where kernels $G^B_1$ and $G^B_2$ are bulk gauge groups for $F_1$ and $F_2$.

Changing the order of integration we obtain (\ref{fd-gl-nd}).
As $h\to 0$ the gluing identity (\ref{fd-gl-nd}) becomes the identity
between formal integrals just as in the non-degenerate case
\[
\int^{formal}_{T_{b_0}B_\pa} Z_{[c_1(b)]} Z_{[c_2(b)]}db=Z_{[c]}
\]
which should be regarded as the contribution of the
critical point $c$ to $Z_F$ written as an iterated integral\footnote{ Recall that $db$ is a $\Gamma_\pa$-invariant
measure on $B_\pa$ such that $\frac{dx}{db} db$ is a $G$-invariant measure on $F$.}.
After a gauge fixing in the integral over $b$ we arrive to the
following formula for the left side:
\begin{multline} \label{glu}
Z_{[c]}=|G_1^B||G_2^B||\Gamma_\pa| (2\pi)^{\frac{N+n}{2}} h^{-\frac{n}{2}}
\frac{\det(-iL_{\varphi_1}(c_1))\det(-iL_{\varphi_2}(c_2))\det(-iL_{\varphi_\pa}(c_\pa))}
{|\det(B_1(c_1))||\det(B_2(c_2))||\det(B_\pa(c_\pa))|} \\ \exp\left(\frac{i}{h}(S_1(c_1)+S_2(c_2))+\frac{i\pi}{4}(\sign(B_1(c_1))+
\sign(B_2(c_2))+\sign(B_\pa(c_\pa)))\right) \\
\left(v_1(c_1)v_2(c_2)v_\pa(c_\pa)+ \sum_{\Ga\neq \O}\mbox{ composite Feynman diagrams }\right),
\end{multline}
Here $N=N_1+N_2-N_\pa=\dim F$ and $n=n_1+n_2-n_\pa$ were $n_i=\dim G_i$ and $n_\pa=\dim \Gamma_\pa$. Composite Feynman diagrams consist of Feynman diagrams for
$F_1$, Feynman diagrams for $F_2$ and Feynman diagrams connecting them
which come from formal integration over boundary fields in the formal neighborhood
of $b_0$. Factors $v_1(c_1), v_2(c_2), v_\pa(c_\pa)$ are densities of corresponding measures 
in local coordinates which we used in (\ref{glu}).

Comparing this expression with (\ref{as-FP}) besides the obvious identity
$S(c)=S(c_1)+S(c_2)$ we obtain identities
\begin{multline}
\frac{\det(-iL_{\varphi_1}(c_1))\det(-iL_{\varphi_2}(c_2))\det(-iL_{\varphi_\pa}(c_\pa))}
{|\det(B_1(c_1))||\det(B_2(c_2))||\det(B_\pa(c_\pa))|} \\ \exp\left(\frac{i\pi}{4}(\sign(B_1(c_1))+
\sign(B_2(c_2))+\sign(B_\pa(c_\pa)))\right)=
\frac{\det(-iL_\varphi(c))}{\sqrt{|\det(B(c))|}}\exp\left(\frac{i\pi}{4}\sign(B(c))\right)
\end{multline}

In addition to this in each order $h^m$ with $m>0$ we will have the identity
\begin{multline}
\sum \mbox{ {\it of all composite Feynman diagrams of order} } m \mbox{ {\it for} } F_1, F_2, B_\pa= \\ \sum \mbox{ {\it of all Feynman diagrams of order} } m \mbox{ {\it for} } F
\end{multline}

These identities are universal algebraic identities which hold for any choice of
of $F_1, F_2$ (as above). This implies that if we define the weights of
Feynman diagrams as prescribed by a path integral, they should
satisfy the same identities and therefore formal semiclassical gluing
partition functions should satisfy the gluing axiom.

\section{Abelian Chern-Simons theory}
In TQFT's there are no ultraviolet divergencies but there is a gauge
symmetry to deal with. Perhaps the simplest non-trivial example of TQFT is the
Abelian Chern-Simons theory with the Lie group $\RR$. Fields in such theory are
connections on the trivial $\RR$-bundle over a compact, smooth, oriented $3$-dimensional manifold $M$.
We will identify fields with $1$-forms on $M$.
The action is
\[
S(A)=\frac{1}{2} \int_M A\wedge dA
\]
Solutions of the Euler-Lagrange equations are closed $1$-forms on $M$.
The variation of this action induces the exact symplectic form  on $\Omega^1(\pa M)$
(see section \ref{cCS}).

\subsection{The classical action and boundary conditions} A choice of metric on $M$ induces a metric on $\pa M$ and the Hodge decomposition:
\[
\Omega(\pa M)=d\Omega(\pa M)\oplus H(\pa M)\oplus d^*\Omega(\pa M)
\]

The Lagrangian subspace
of boundary values of solutions to Euler-Lagrange equations is
\[
L_M=H^1_M(\pa M)\oplus d\Omega^0(\pa M)
\]
where $H_M(\pa M)$ is the space of harmonic representatives of cohomology classes on the boundary
coming from cohomology classes $H^1(M)$ of the bulk by pull-back with respect to inclusion of the boundary.

Choose a decomposition of $H(\pa M)$ into a direct sum of two Lagrangian
subspaces:
\[
H(\pa M)=H_+(\pa M)\oplus H_-(\pa M)
\]
This induces a decomposition of forms $\Omega(\pa M)=\Omega_+(\pa M)\oplus \Omega_-(\pa M)$
where
\[
\Omega_+(\pa M)=H_+(\pa M) \oplus d\Omega(\pa M), \ \ \Omega_-(\pa M)=H_-(\pa M) \oplus d^*\Omega(\pa M)
\]
Choose the boundary Lagrangian fibration as
\[
p_\pa: \Omega(\pa M)\to B(\pa M)=\Omega_+(\pa M)
\]
with fibers
\[
p_\pa^{-1}(b)=b+\Omega_-(\pa M)\simeq H_-(\pa M)\oplus d^*\Omega(\pa M) .
\]

This fibration is not $\alpha_{\pa M}$-exact, i.e. the restriction of $\alpha_{\pa M}$ to fibers
is zero. Let us modify the action, by adding a boundary term such that the
form $\alpha_{\pa M}$ will vanish on fibers of $p$. Define the new action as
\[
\widetilde{S}(A)=S(A)+\frac{1}{2} \int_{\pa M} A_+\wedge A_-
\]
where $A_\pm$ are $\Omega_\pm$-components of $i^*(A)$.

The new form on boundary connections is
\[
\widetilde{\alpha}_{\pa M}(a)=\alpha_{\pa M}(a)+\frac{1}{2}\delta  \int_{\pa M} a_+\wedge a_-=-\int_{\pa M} a_-\wedge \delta a_+
\]
and it vanishes on the fibers of $p_\pa$ because on each fiber $\delta a_+=0$.

Note that the modified action is gauge invariant. Indeed, on components $A_\pm$ gauge transformations act as
$A_+\mapsto A_+ +d\theta$ and $A_-\mapsto A_-$, i.e. gauge transformations act trivially on fibers.

\subsection{Formal semiclassical partition function}

\subsubsection{More on boundary conditions}For this choice of Lagrangian fibration the bulk gauge group $G^B$
is $\Omega^0(M,\pa M)$. The boundary gauge group acts trivially on fibers. Indeed, the boundary gauge group $\Omega^0(\pa M)$ acts naturally on the base $B(\pa M)=H^1(\pa M)_+\oplus
d\Omega^0(\pa M)$, $\alpha\mapsto \alpha+ d\lambda$. It acts on the
base shifting the fibers: $p(\beta+d\lambda)=p(\beta)+d\lambda$.

According to the general scheme outlined in section \ref{fdg}, in order to define the
formal semiclassical partition function we have to fix a
background flat connection $a$ and ``integrate" over the fluctuations
$\sqrt{h}\alpha$ with boundary condition $i^*(\alpha)_+=0$. We have

\[
\widetilde{S}(a+\alpha)=\tilde{S}(\alpha)
+\frac{1}{2}\int_{\pa M} a_+\wedge a_-
\]
Note that $da=0$ which means that $a$ restricted to the boundary is a closed form
which we can write as $i^*(a)=[a]_++[a]_-+d\theta$ where $[a]_\pm\in H_\pm(\pa M)$.
Therefore, for the action we have:
\[
\widetilde{S}(a+\alpha)=\widetilde{S}(\alpha)
+\frac{1}{2}<[a]_+, [a]_->_{\pa M}
\]
where $<.,.>$ is the symplectic pairing in $H(\pa M)$.

For semiclassical quantization we should choose the gauge fixing submanifold $\Lambda\subset \Omega(M)$, such that $(T_aF_M)_+=T_aEL\oplus T_a\Lambda$. Here $(T_aF_M)_+$ is the space of $1$-forms ($\alpha$-fields) with boundary
condition $i^*(\alpha)_+=0$. As it is shown in Appendix \ref{hodge}
the action functional restricted to fields with boundary values in an
isotropic subspace $I\subset \Omega^1(\pa M)$ is non-degenerate on
\[
T_a\Lambda_I =d^*\Omega^2_N(M, I^\perp)\cap \Omega^1_D(M, I)
\]
For our choice of boundary conditions $I=\Omega^1_-(\pa M)$.

\subsubsection{Closed space time} First, assume the space time has no boundary. Then the formal semiclassical
partition function is defined as the product of determinants which arise
from gauge fixing and from the Gaussian integration as in (\ref{as-FP}).
In the case of Abelian Chern-Simons the gauge condition is $d^*A=0$
and the action of the gauge Lie algebra $\Omega^0(M)$ on
the space of fields $\Omega^1(M)$ is given by the map $d: \Omega^0(M)\to
\Omega^1(M)$ (here we identified $\Omega^1(M)$ with its tangent space at
any point). Thus, the FP action (\ref{fp-action}) in our case is
\[
S(A, \overline{c}, c, \lambda)=\frac{1}{2}\int_M A\wedge dA +\int_M \overline{c} \ \ \Delta c \ \ d^3x+
\int_M \lambda \ \ d^*A  \ \ d^3x
\]
where $\overline{c}, c$ are ghost fermion fields, and $\lambda$ is the Lagrange multiplier
for the constraint $d^*A=0$.

By definition the corresponding Gaussian integral is
\[
Z_a=C\frac{|{\det}'(\Delta_0)|}{\sqrt{|{\det}'(\widehat{\ast d})|}}\exp(\frac{i\pi}{4}(2 \text{sign}(\Delta_0)+\text{sign}(\widehat{\ast d})))
\]
Here ${\det}'$ is a regularized determinant and $ \text{sign}(A)$ is the signature of the
differential operator $A$. The constant depends of the choice of regularization. The
usual choice is the $\zeta$-regularization. The signature is up to a normalization
the eta invariant \cite{W}. The operator $\widehat{\ast d}$ acts on $\Omega^1(M)\oplus \Omega^0(M)$
as
\begin{equation}\label{d-hat}
\left(
\begin{array}{cc}
\ast d &  d \\ d^*  & 0
\end{array}
\right)
\end{equation}
Its square is the direct sum of Laplacians:
\[
\widehat{\ast d}^2=\left(
\begin{array}{cc}
d^*d+dd^* &  0 \\ 0  & d^*d
\end{array}
\right)
\]
Thus the regularized determinant of $\widehat{\ast d}$ is the product of determinants
acting on 1-forms and on 0-forms:
\[
|{\det}'(\widehat{\ast d})|^2=|{\det}'(\Delta_1)||{\det}'(\Delta_0)|
\]
This gives the following formula for the determinant contribution to the partition function:
\begin{equation}\label{z-1}
\frac{|{\det}'(\Delta_0)|}{\sqrt{|{\det}'(\widehat{\ast d})|}}=\frac{|{\det}'(\Delta_0)|^{\frac{3}{4}}}{|{\det}'(\Delta_1)|^{\frac{1}{4}}}
\end{equation}
Taking into account that $\ast \Omega^i(M)=\Omega^{3-i}(M)$ we can write this as
\[
T^{1/2}=|{\det}'(\Delta_1)|^{\frac{1}{4}}|{\det}'(\Delta_2)|^{\frac{2}{4}}|{\det}'(\Delta_3)|^{\frac{3}{4}}
\]
where $T$ is the Ray-Singer torsion. This gives well-known formula
for the absolute value of the partition function of the Abelian Chern-Simons theory on a closed manifold.
\begin{equation}\label{part-torsion}
|Z|=CT^{1/2}
\end{equation}
We will not discuss here the $\eta$-invariant part.

\begin{remark} The operator $\widehat{\ast d}$ is easy to identify with $L_-=\ast d+ d\ast$, acting
on $\Omega^1(M)\oplus \Omega^3(M)$ from \cite{W}. Indeed, using Hodge star we can identify
$\Omega^0(M)$ and $\Omega^3(M)$. After this the operators are related as
\[
L_-= \left(
\begin{array}{cc}
1 &  0 \\ 0  & \ast
\end{array}
\right) \widehat{ \ast d} \left(
\begin{array}{cc}
1 &  0 \\ 0  & \ast
\end{array}
\right)^{-1}
\]
\end{remark}

\begin{remark} There is one more formula in the literature for gauge fixing.
Assume that a Lie group $G$ has an invariant inner product,
the space of fields $F$ is a Riemannian manifold and $G$ acts
by isometries on $F$.
In this case there is a natural gauge fixing
which leads to the following formula for an integral of a $G$-invariant function \cite{Sch1}:
\[
\int_F h(x) dx=|G|\int_{F/G} h(x) ({\det}'(\tau_x^*\tau_x))^{\frac{1}{2}} [dx]
\]
Here we assume that the $G$-action does not have stabilizers. The linear mapping $\tau_x: \g \to
T_xF$ is given by the $G$-action, the Hermitian conjugate is taken with respect to the
metric structure on $F$ and on $G$, $dx$ is the Riemannian volume on $F$ and $[dx]$ is
the Riemannian volume on $F/G$ with respect to the natural Riemannian structure on the quotient space.

For the Abelian Chern-Simons a choice of metric on the space time induces metrics on $G=\Omega^0(M)$ and
on $F=\Omega^1(M)$. The gauge group $G$ acts on $F$ by isometries and $\tau_x=d$, the de Rham differential.
This gives another expression for the absolute value of the partition function
\begin{equation}\label{z-3}
|Z|=C\frac{|{\det}'(\Delta_0)|^{\frac{1}{2}}}{|{\det}'(\ast d)|^{\frac{1}{2}}}
\end{equation}
Here $\ast d: \Lambda\to \Lambda$, and $\Lambda=d^*\Omega^2(M)$ is the submanifold on
which the action functional is non-degenerate. It is clear that this formula coincides with (\ref{part-torsion}).

\end{remark}

\subsubsection{Space time with boundary}
Now let us consider the case when $\pa M$ is non-empty.
In this case the bulk gauge group $G^B$ is $\Omega_D(M, \{0\})$ which we will denote just $\Omega_D(M)$.
The space of fluctuations is $\Omega^1_D(M, \Omega_-(\pa M))$. The bilinear from in the Faddeev-Popov action
is
\[
\frac{1}{2}\int_M \alpha\wedge d\alpha +\int_M \lambda \ d^*\alpha \ d^3x -i\int_M \overline{c} \ \Delta c \ d^3x
\]
The even part of this form is symmetric if we impose the boundary condition $i^*(\lambda)=0$.
Similarly to the case of closed space time we can define the partition function
as

\begin{equation}\label{pf-acs}
Z_{a,M}=C|{\det}'(\widehat{\ast d)}|^{-1/2} |{\det}'( \Delta_0^{D,\{0\}})| \exp(\frac{i\pi}{4}(2\, \text{sign}(\Delta_0)+\text{sign}(\widehat{\ast d}))) \exp(\frac{i}{h}<[a]_+, [a]_->_{\pa M})
\end{equation}
Here $\Delta_0^{D,\{0\}}$ is the Laplace operator action on $\Omega_D(M, \{0\})$ and $[a]_\pm$ are
the $\pm$ components of the cohomology class of the boundary value $i^*(a)$ of $a$. The operator
$\widehat{\ast d}$ acts on $\Omega^1_D(M, \Omega_-(\pa M))\oplus \Omega_D^0(M, \{0\})$ and is
given by (\ref{d-hat}). This ratio of determinants is expected to give a version of the Ray-Singer torsion for appropriate boundary conditions. The signature contributions are expected to be the $\eta$-invariant with the
appropriate boundary conditions. For the usual choices of boundary conditions, such as tangent, absolute,
or APS boundary conditions at least some of these relations are known, for more general boundary conditions
it is a work in progress.


\subsubsection{Gluing} According to the finite dimensional gluing formula we expect
a similar gluing formula for the partition function.
A consequence of this formula is the multiplicativity of the version of the Ray-Singer
torsion with boundary conditions
described above. To illustrate this, let us take a closer look at the exponential part of (\ref{pf-acs}).

Recall that $L_M\subset \Omega^1(\pa M)$ is
the space of closed 1-forms which are boundary values of closed 1-forms on $M$.
To fix boundary conditions we fixed the decomposition $\Omega^1(\pa M)=\Omega^1(\pa M)_+\oplus \Omega^1(\pa M)_-$
(see above).

Let $\beta$ be a tangent vector to $L_M$ at the point $i^*(a)\in L_M$. We have natural identifications
\[
T_{i^*(a)}\Omega^1(\pa M)_-=H^1(\pa M)_-\oplus d^*\Omega^2(\pa M) ,\ \ T_{i^*(a)}\Omega^1(\pa M)_+=H^1(\pa M)_+\oplus d\Omega^0(\pa M)
\]
Denote by $\beta_\pm$ the components of $\beta$ in $T_{i^*(a)}L_\pm$ respectively.
Since $d\beta=0$ we have $\beta_+=[\beta]_++d\theta$, and $\beta_-=[\beta]_-$, where $[\beta]_\pm$
are components of the cohomology $[\beta]_\pm$ in $H^1(\pa M)_\pm$.
If the reduced tangent spaces $[T_{i^*(a)}L_M]=H^1_\pm(\pa M)$ and $[T_{i^*(a)}\Omega^1(\pa M)_\pm]=H^1_M(\pa M)$
are transversal, which is what we assume here, projections to $[T_{i^*(a)}\Omega^1(\pa M)_\pm]$ give
linear isomorphisms $A^{(\pm)}_M: H^1_M(\pa M)\to H^1_\pm(\pa M)$.
This defines the linear isomorphism
\[
B_M=A^{(-)}_M(A^{(+)}_M)^{-1}:  H^1(\pa M)_+\to  H^1(\pa M)_-
\]
acting as $B_M([\beta]_+)=[\beta]_-$ for each $[\beta] \in H^1_M(\pa M)$.
This is the analog of the Dirichlet-to-Neumann operator.

Now considering small variations around $a$ have
\begin{multline}
Z_{[a+\sqrt{h}\beta]}=Z_{[a]}\exp(\frac{i}{\sqrt{h}}(<[i^*(a)]_+, B_M([i^*(\beta)]_+)>_{\pa M}+\\
<[i^*(\beta)]_+, [i^*(a)]_->_{\pa M})+ i<[i^*(\beta)]_+, B_M([i^*(\beta)]_+)>_{\pa M})
\end{multline}

The gluing formula for this semiclassical partition function at the level of exponents gives the gluing formula for
Hamilton-Jacobi actions. At the level of pre-exponents it also gives the gluing formula for torsions and for the $\eta$-invariant for appropriate boundary conditions. Changing boundary conditions results in a boundary contribution
to the partition function and to the gluing identity.
One should also expect the gluing formula for correlation functions. The details
of these statements require longer discussion and substantial analysis and will be done
elsewhere.

There are many papers on Abelian Chern-Simons theory. The appearance of torsions and $\eta$-invariants
in the semiclassical asymptotics of the path integral for the Chern-Simons action was first
pointed out in \cite{W}. For a geometric approach to compact
Abelian Chern-Simons theory and a discussion of gauge fixing and the appearance of torsions
in the semiclassical analysis see \cite{Manu}. For the geometric quantization approach
to the Chern-Simons theory with compact Abelian Lie groups see \cite{An}.

\appendix

\section{Discrete time quantum mechanics}\label{d-time-qm} An example of a finite dimensional
version of a classical field theory is a discrete time approximation to the
Hamiltonian classical mechanics of a free particle on $\RR$. We denote coordinates on this space $(p,q)$
where $p$ represents the momentum and $q$ represents the coordinate of the system.

In this case the space time is an ordered collection of $n$ points which represent the discrete
time interval. If we enumerate these points $\{1,\dots,n\}$ the points $1,n$ represent the
boundary of the space time. The space of fields is $\RR^{n-1}\times \RR^n$ with coordinates $p_i$ where $i=1,\dots, n-1$ represents the ``time interval'' between points $i$ and $i+1$ and $q_i$ where $i=1,\dots, n$. The coordinates
$p_1, p_{n-1}, q_1, q_n$ are boundary fields\footnote{ In other words the space time is a $1$-dimensional cell complex. Fields assign coordinate function $q_i$ to the vertex $i$ and $p_i$ to the edge $[i,i+1]$.}. The action is
\[
S=\sum_{i=1}^{n-1} p_i(q_{i+1}-q_i)-\sum_{i=1}^{n-1} \frac{p_i^2}{2}
\]

We have
\[
dS=\sum_{i=1}^{n-2} (q_{i+1}-q_i-p_i)dp_i+\sum_{i=2}^{n-1}(p_{i-1}-p_i)dq_i+
p_{n-1}dq_n- p_1dq_1
\]
From here we derive the Euler-Lagrange equations
\[
q_{i+1}-q_i=p_i, \ \ i=1,\dots n-1,
\]
\[
p_{i-1}-p_i=0, \ \ i=2,\dots, n-1
\]
and the boundary 1-form
\[
\alpha=p_{n-1}dq_n-p_1dq_1
\]
This gives the symplectic structure on the space of boundary fields
with
\[
\omega_\pa= dp_{n-1}\wedge dq_n-dp_1\wedge dq_1
\]
The boundary values of solutions of the Euler-Lagrange equations
define the subspace
\[
L=\pi(EL)=\{(p_1,q_1,p_{n-1},q_n)| p_1=p_{n-1}, q_n=q_1+(n-1) p_1\}
\]
It is clear that this a Lagrangian subspace.

\section{Feynman diagrams}\label{Fd}
Let $\Gamma$ be a graph with vertices of valency $\geq 3$ with one special vertex which
may also have valency $0,1,2$. To define the weight $F_c(\Gamma)$,
cut $\Gamma$ into the union of stars of vertices and edges. Denote the result by $\hat{\Gamma}$,
see an example on Fig. \ref{theta}.

\begin{figure}[htb]
\includegraphics[height=2.5cm,width=8cm]{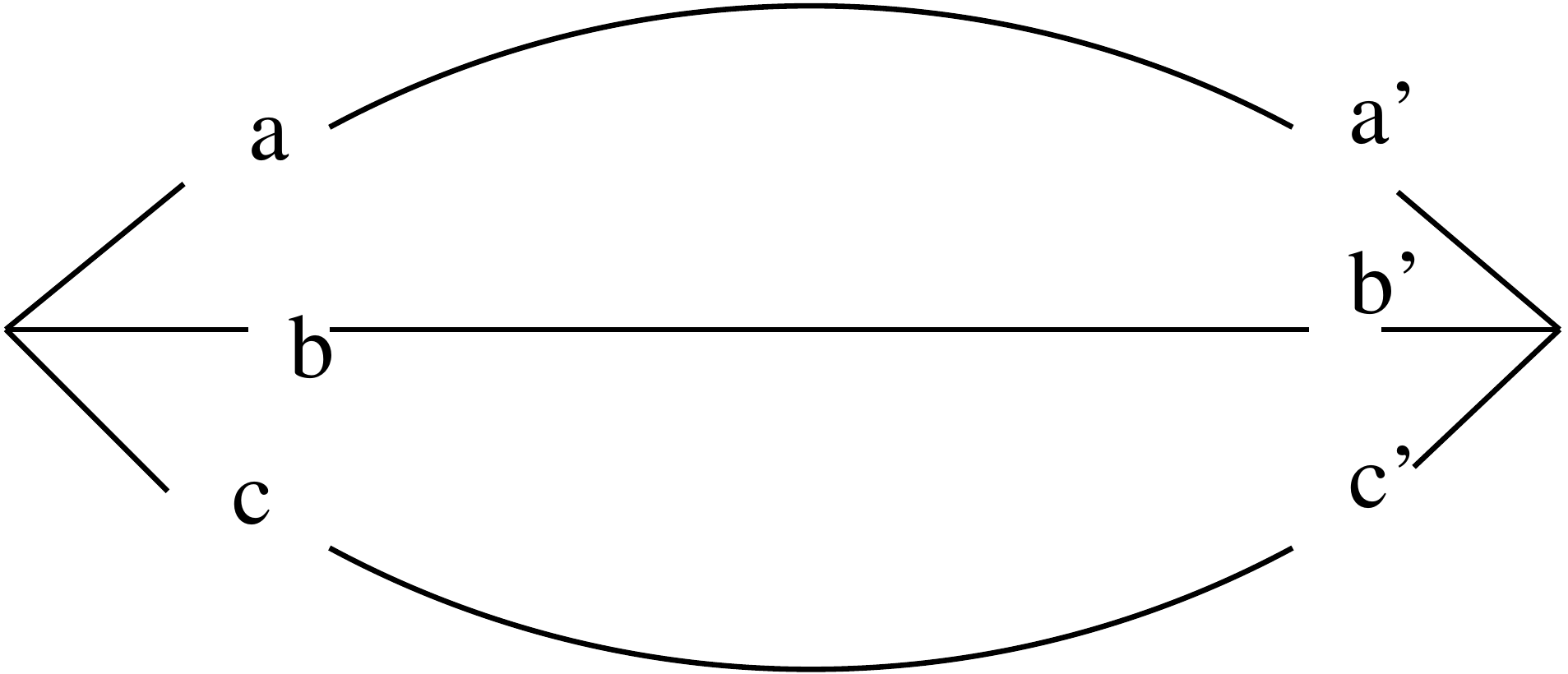}
\caption{The "theta" diagram.}
\label{theta}
\end{figure}


A state on $\hat{\Gamma}$ is a mapping from endpoints of stars and edges of $\hat{\Gamma}$
to the set $1,\dots, n$, for an example see Fig. \ref{theta}.
The weight of $\Gamma$ is defined as
\[
F_c(\Gamma)=\sum_{states}  \frac{\pa^l v}{\pa x^{j_1}\dots \pa x^{j_l}}(c) \prod_{vertices} \frac{\pa^k S}{\pa x^{i_1}\dots \pa x^{i_k}}(c)\prod_{edges} (B_c^{-1})_{ij}
\]
Here the sum is taken over all states on $\hat{\Gamma}$, and $i_1,\dots, i_k$ are states on the endpoints of edges in the start of a vertex. The first factor is the weight of the special vertex where $v$ is the density of the integration measure in local coordinates $\frac{dx}{db}=v(x) dx^1\dots dx^N$.
The $(i,j)$ is the state at endpoints of an edge. Note that weights of vertices and the matrix $B_c$
are symmetric. This makes the definition meaningful.

\section{Gauge fixing in Maxwell's electromagnetism}

In the special case of electromagnetism ($G = \RR, \mathfrak{g} = \mathbb{R}$), the space of fields is $F_M = \Omega^1(M) \oplus \Omega^{n-2}(M)$ and similarly for the boundary. If $M$ has no boundary, the gauge group $G_M = \Omega^0(M)$ acts on fields as follows: $A \mapsto A + d \alpha, B \mapsto B$. We can construct a global section of the corresponding quotient using Hodge decomposition: we know that

\begin{equation} \Omega^{\bullet}(M) \cong \Omega_{\text{exact}}^{\bullet}(M) \oplus H^{\bullet}(M) \oplus \Omega_{\text{coexact}}^{\bullet}(M) \end{equation}
where the middle term consists of harmonic forms. In particular,

\begin{equation} \Omega^1(M) = d \Omega^0(M) \oplus H^1(M) \oplus d^{\ast} \Omega^2(M) \end{equation}
where the last two terms give a global section. In physics, choosing a global section is called gauge fixing, and this particular choice of gauge is called the Lorentz gauge, where $d^{\ast} A = 0$.

\section{Hodge decomposition for Riemannian manifolds with boundary}\label{hodge}

\subsection{Hodge decomposition with Dirichlet and Neumann boundary conditions} Let $M$ be a smooth oriented Riemannian manifold with boundary $\pa M$.
Recall some basic facts about the Hodge decomposition of differential forms on $M$.
Choose local coordinates near the boundary in which the metric has the product structure
with $t$ being the coordinate in the normal direction.
Near the boundary any smooth form can be written as

\[
\omega=\omega_{tan}+\omega_{norm}\wedge dt
\]
where $\omega_{tan}$ is the tangent component of
$\omega$ near the boundary and $\omega_{norm}$ is the normal component.

We will denote by $\Omega_D(M)$ the space of forms satisfying the Dirichlet
boundary conditions $\io^*(\omega)=0$ where $\io^*$ is the pull-back of the form
$\omega$ to the boundary. This condition
can be also written as $\omega_{tan}=0$.

We will denote by $\Omega_N(M)$ the space of forms satisfying the Neumann
boundary conditions $\io^*(\ast\omega)=0$. Here $\ast: \Omega^i(M)\to \Omega^{n-i}(M)$
is the Hodge star operation, recall that $\ast^2=(-1)^{i(n-i)}\mathrm{id}$ on $\Omega^i(M)$.
Because $\omega_{norm}=\ast'\io^*(\ast\omega )$
the Neumann boundary condition can be written as $\omega_{norm}=0$.

Denote by $d^*=(-1)^i{\ast}^{-1} d \ast$ the formal adjoint of $d$, and by $\Delta=dd^*+d^* d$
the Laplacian on $M$.
Denote by $\Omega_{cl}(M)$ closed forms on $M$, $\Omega_{ex}(M)$ exact forms
on, $\Omega_{cocl}(M)$ the space of coclosed forms, i.e. closed with respect to $d^*$
and by $\Omega_{coex}(M)$ the space of coexact forms.

Define subspaces:
\[
\Omega_{cl,cocl}(M)=\Omega_{cl}(M)\cap \Omega_{cocl}(M), \ \ \Omega_{cl,coex}(M)=\Omega_{cl}(M)\cap \Omega_{coex}(M)
\]
and similarly $\Omega_{ex,cocl}(M)$, $\Omega_{cl,cocl, N}(M)$ and $\Omega_{cl,cocl, D}(M)$.

\begin{theorem} 1) The space of forms decomposes as
\[
\Omega(M)=d^* \Omega_N(M)\oplus \Omega_{cl,cocl}(M) \oplus d\Omega_D(M)
\]

2) The space of closed, coclosed forms decomposes as
\[
\Omega_{cl, cocl}(M)=\Omega_{cl,cocl, N}(M)\oplus \Omega_{ex,cocl}(M)
\]
\[
\Omega_{cl, cocl}(M)=\Omega_{cl,cocl, D}(M)\oplus \Omega_{cl,coex}(M)
\]
\end{theorem}

We will only outline the proof of this theorem. For more details and references on the Hodge decomposition for manifolds with boundary and Dirichlet and Neumann boundary conditions see \cite{CTGM}. Riemannian structure on $M$ induces the scalar product on forms
\begin{equation}\label{sc-prod}
(\omega, \omega')=\int_M \omega\wedge\ast\omega'
\end{equation}
For two forms of the same degree  we have $\omega(x)\wedge \ast\omega'(x)=<\omega(x),\omega'(x)>dx$ where $dx$ is the Riemannian volume form and $<.,.>$ is the scalar product on
$\wedge^kT^*_xM$ induced by the metric. This is why (\ref{sc-prod}) is positive definite.

\begin{lemma}\label{orth-cocl} With respect to the scalar product (\ref{sc-prod})
\[
(d\Omega_D(M))^\perp=\Omega_{cocl}
\]
\end{lemma}
\begin{proof} By the Stokes theorem for any form $\theta\in \Omega^{i-1}_D(M)$
we have
\[
(\omega, d\theta)=\int_M \omega\wedge \ast d\theta= (-1)^{(i+1)(n-i)}(\int_{\pa M} \io^*(\ast \omega)\wedge \io^*(\theta)+ \int_M d\ast \omega\wedge \theta)
\]
The boundary integral is zero because $\theta\in \Omega_D(M)$. Thus $(\omega, d\theta)=0$
for all $\theta$ if and only if $d\ast \omega=0$ which is equivalent to $\omega\in \Omega_{cocl}(M)$.
\end{proof}

\begin{corollary} Because $d\Omega_D(M)\subset \Omega_{cl}(M)$, we have
$\Omega_{cl}(M)=\Omega_{cl}(M)\cap (d\Omega_D(M))^\perp\oplus d\Omega_D(M)$. i.e.
\[
\Omega_{cl}(M)=\Omega_{cl,cocl}(M)\oplus d\Omega_D(M)
\]
\end{corollary}

Here we are sketchy on the analytical side of the story. If $U\subset V$ is a subspace in an
inner product space, in the infinite dimensional setting more analysis might be
required to prove that $V=U\oplus U^\perp$. Here and below we just assume that this
does not create problems. Similarly to Lemma \ref{orth-cocl} we obtain
\[
(d^*\Omega_N(M))^\perp=\Omega_{cl}(M)
\]
This completes the sketch of the proof of the first part. The proof of the second part
is similar.

Note that the spaces in the second part of the theorem are harmonic forms representing
cohomology classes:
\[
\Omega_{cl,cocl, N}(M)=H(M), \ \ \Omega_{cl,cocl, D}(M)=H(M,\pa M)
\]

\subsection{More general boundary conditions}
\subsubsection{} Assume that $M$ is a smooth compact
Riemannian manifold, possibly with non-empty boundary $\pa M$. Let $\pi: \Omega^i(M)\to \Omega^i(\pa M)$, $i=0, \dots, n-1$ be the
restriction map (the pull-back of a form to the boundary) and $\pi(\Omega^n(M))=0$.

The Riemannian structure on $M$ induces the metric on $\pa M$. Denote by $\ast$ the Hodge star for $M$, and by $\ast_\pa$ the Hodge star
for the boundary $\ast_\pa : \Omega^i(\pa M)\to \Omega^{n-1-i}(\pa M)$.
Define the map $\widetilde{\pi}: \Omega(M)\to \Omega(\pa M)$, $i=1,\dots, n$ as the composition $\widetilde{\pi}(\alpha)=\ast_\pa \pi (\ast\alpha)$. Note that $\tilde{\pi}(\Omega^0(M))=0$.

Denote by $\Omega_D(M, L)$ and $\Omega_N(M,L)$ the following subspaces:
\[
\Omega_D(M, L)=\pi^{-1}(L), \ \ \Omega_N(M, L)=\widetilde{\pi}^{-1}(L)
\]
where $L\subset \Omega(\pa M)$ is a subspace.

Denote by $L^\perp$ the orthogonal complement to $L$ with respect to the Hodge inner product on the boundary.
The following is clear:
\begin{lemma}
\[
(\ast L^{(i)})^\perp=\ast(L^{(i)})^\perp, \ \ \ast (L^\perp)=L^{sort}
\]
Here $L^{sort}$ is the space which is symplectic orthogonal to $L$.
\end{lemma}
\begin{proposition} $(d^*\Omega_N(M,L))^\perp=\Omega_D(M,L^\perp)_{cl}$
\end{proposition}
\begin{proof} Let $\omega$ be an $i$-form on $M$ such that
\[
\int_M \omega\wedge d\ast \alpha=0
\]
for any $\alpha$. Applying Stocks theorem we obtain
\[
\int_M \omega\wedge d\ast \alpha= (-1)^i \int_{\pa M} \pi(\omega)\wedge \ast_\pa \tilde{\pi}(\alpha)+ (-1)^{i+1} \int_M d\omega\wedge \ast\alpha
\]
The boundary integral is zero for any $\alpha$ if and only if $\pi(\omega)\in L^\perp$ and the bulk integral is
zero for any $\alpha$ if and only if $d\omega=0$.
\end{proof}

As a corollary of this we have the orthogonal decomposition
\[
\Omega(M)=\Omega_D(M,L^\perp)_{cl}\oplus d^*\Omega_N(M,L)
\]
Similarly, for each subspace $L\subset \Omega(\pa M)$ we have
the decomposition
\[
\Omega(M)=\Omega_N(M,L^\perp)_{cocl}\oplus d\Omega_D(M,L)
\]
Now, assume that we have two subspaces $L, L_1\subset \Omega(\pa M)$
such that
\begin{equation}\label{ass}
d_\pa (L_1^{\perp})\subset L^\perp,
\end{equation}
Note that this implies $d^*_\pa L\subset L_1$.
Indeed, fix $\alpha \in L$, then (\ref{ass}) implies that
for any $\beta\in L_1^\perp$ we have
\[
\int_{\pa M} \alpha\wedge \ast d_\pa \beta=0
\]
This is possible if and only if
\[
\int_{\pa M} \ast d_{\pa} \ast \alpha\wedge \ast \beta=0
\]
Thus, $d_\pa^*\alpha\in L_1$. Here we assumed that $(L_1^\perp)^\perp=L_1$.

Because $\pi d=d_\pa \pi$ and $\tilde{\pi} d^*=d^*_\pa \tilde{\pi}$ we also have
\[
d\Omega_D(M,L_1^\perp)\subset \Omega_D(M,L^\perp)_{cl}, \ \ d^*\Omega_N(M,L)\subset \Omega_N(M,L_1)_{cocl}
\]

\begin{theorem} Under assumption (\ref{ass}) we have
\begin{equation}\label{H-decomp}
\Omega(M)=d^*\Omega_N(M,L)\oplus \Omega_D(M,L^\perp)_{cl}\cap \Omega_N(M,L_1)_{cocl}\oplus d\Omega_D(M,L_1^\perp)
\end{equation}
\end{theorem}

Indeed, if $V, W\subset \Omega$ are liner subspaces in the scalar product space $\Omega$
such that $W\subset V^\perp$ and $V\subset W^\perp$ then $\Omega=V\oplus V^\perp=W\oplus W^\perp$ and
\[
\Omega=V\oplus W^\perp\cap V^\perp \oplus W
\]

We will call the identity (\ref{H-decomp}) the Hodge decomposition with
boundary conditions. The following is clear:

\begin{theorem} The decomposition (\ref{H-decomp}) agrees with the Hodge star
operation if and only if
\[
\ast L_1^\perp=L
\]
\end{theorem}

\begin{remark}In the particular case $L=\{0\}$ and $L_1^\perp=\{0\}$ we obtain
the decomposition from the previous section:
\[
\Omega(M)=d^*\Omega_N(M)\oplus \Omega_{cl,cocl}(M)\oplus d\Omega_D(M)
\]
\end{remark}
\begin{lemma} If $L\subset \Omega(\pa M)$ is an isotropic subspace
then $\ast L\subset \Omega(\pa M)$ is also an isotropic subspace.
\end{lemma}
Indeed, if $L$ is isotropic then for any $\alpha, \beta\in L$ we have
$\int_{\pa M} \alpha\wedge \ast\beta=0$, but
\[
\int_{\pa M} \ast \alpha\wedge \ast^2\beta=\pm\int_{\pa M} \alpha\wedge\ast\beta
\]
therefore $\ast L$ is also isotropic.

\begin{remark} We have
\[
\ast \Omega_N(M)=\Omega_D(M),  \ \ \ast H(M)=H(M,\pa M)
\]
In the second formula $H(M)$ is the space of closed-coclosed forms with Neumann boundary conditions
and $H(M, \pa M)$ is the space of closed-coclosed forms with Dirichlet boundary conditions. They
are naturally isomorphic to corresponding cohomology spaces. Note that as a consequence of the first identity we have $*d^*\Omega_N(M)=d\Omega_D(M)$.
We also have more general identity
\[
\ast \Omega_N(M, L)=\Omega_D(M, \ast_\pa L)
\]
and consequently $\ast \Omega_D(M,L)=\Omega_N(M,\ast_\pa L)$.
\end{remark}

Let $\pi$ and $\tilde{\pi}$ be maps defined at the beginning of this section.
Because $\pi$ commutes with de Rham differential and $\tilde{\pi}$ commutes
with its Hodge dual, we have the following proposition
\begin{proposition} Let $H_M(\pa M)$ be the space of harmonic forms
on $\pa M$ extendable to closed forms on $M$, then
\[
\pi(\Omega_{cl}(M))=H_M(\pa M)\oplus d\Omega(\pa M), \ \ \widetilde{\pi}(\Omega_{cocl}(M))=H_M(\pa M)^\perp\oplus d^*\Omega(\pa M)
\]
\end{proposition}
Here is an outline of the proof. Indeed, let $\theta\in \Omega_{cl}(M)$ and $\sigma \in \Omega_{cocl}(M)$. Then
\[
\int_{\pa M} \pi(\theta)\wedge \ast_\pa \widetilde{\pi}(\sigma)= \int_{\pa M} \pi(\theta)\wedge\pi(\ast \sigma)=
 \int_M d(\theta\wedge \ast\sigma)
\]
The last expression is zero because by the assumption $\theta$ and $\ast \sigma$ are closed.
The proposition follows now from the Hodge decomposition for forms on the boundary
and from $\pi(\Omega_{cl}(M))\subset \Omega_{cl}(\pa M)$, $\widetilde{\pi}(\Omega_{cocl}(M))\subset \Omega_{cocl}(\pa M)$.

\subsubsection{$\dim M=3$} Let us look in details
at the 3-dimensional case. In order to have the Hodge decomposition with boundary
conditions we required
\[
d L_1^\perp\subset L^\perp
\]
If we want it to be {\it invariant with respect to the Hodge star } we should also have
$\ast L_1^\perp=L$. Together these two conditions imply that $L$ should
satisfy $d\ast L\subset L^\perp$ or
\[
\int_{\pa M} d\ast \alpha\wedge \ast \beta =0
\]
for any $\alpha, \beta \in L$. This condition is equivalent to
\[
\int_{\pa M} d^* \alpha\wedge \beta =0
\]
for any $\alpha\in L^{(1)}$ and any $\beta\in L^{(2)}$.

Note that if $L^{(2)}=\{0\}$ we have no conditions on the subspace $L^{(1)}$.
In this case for any choice of $L^{(0)}$ and $L^{(1)}$  the $\ast$-invariant
Hodge decomposition  is:
\[
\Omega^0(M)=d^*\Omega^1_N(M,L^{(0)})\oplus \Omega^0_D(M, {L^{(0)}}^\perp)_{cl}
\]
\[
\Omega^1(M)=d^*\Omega^2_N(M,L^{(1)})\oplus \Omega^1_D(M, {L^{(1)}}^\perp)_{cl}\cap
\Omega^1_N(M, L^{(0)})_{cocl}\oplus d\Omega_D^0(M)
\]
Here we used $\Omega_N^i(M,L_1)=\Omega^i_N(M, L_1^{(i-1)})=\Omega^i_N(M, (\ast L^{(3-i)})^\perp)$.
The condition $L^{(2)}=\{0\}$ implies that $\Omega^1_N(M,(\ast L^{(2)})^\perp)=
\Omega^1(M)$.
We also used $\Omega^0_D(M, L^\perp_1)=\Omega^0(M, \ast L^{(2)})=\Omega^0_D(M)$.

The decomposition of $2$- and $3$ -forms is the result of application of Hodge star
to these formulae.

\subsubsection{} Consider the bilinear form
\begin{equation}\label{B-form}
B(\alpha, \beta)=\int_M \beta\wedge d\alpha
\end{equation}
on the space $\Omega^\bullet(M)$.

Let $I\subset \Omega^\bullet(\pa M)$ be an isotropic subspace.
\begin{proposition} The form $B$ is symmetric on the space $\Omega_D(M,I)$.
\end{proposition}
Indeed
\[
\int_M(\beta\wedge d\alpha)=(-1)^{|\beta|+1}\int_{\pa M}\pi(\beta)\wedge \pi(\alpha)+
\int_M d\beta\wedge \alpha =(-1)^{(|\alpha|+1)(|\beta|+1)}B(\alpha, \beta)
\]
The boundary term vanishes because boundary values of $\alpha$ and $\beta$ are
in an anisotropic subspace $I$.

\begin{proposition} Let $I\subset \Omega(\pa M)$ be an isotropic subspace, then $B$ is nondegenerate on
$d^*\Omega_N(M,I^\perp)\cap \Omega_D(M,I)$.
\end{proposition}
\begin{proof} If $I$ is isotropic, $\beta\in \Omega_D(M,I)$ and $B(\beta, \alpha)=0$ for any $\alpha\in \Omega_D(M,I)$,
we have:
\[
B(\beta, \alpha)=B(\alpha, \beta)=\int_M \alpha\wedge d\beta
\]
and therefore $d\beta=0$. Therefore $\Omega_D(M,I)_{cl}$ is the kernel of the
form $B$ on $\Omega_D(M,I)$. But we have the decomposition
\[
\Omega(M)=\Omega_D(M,I)_{cl}\oplus d^*\Omega_N(M,I^\perp)
\]
This implies
\[
\Omega_D(M,I)= \Omega_D(M,I)_{cl}\oplus d^*\Omega_N(M,I^\perp)\cap \Omega_D(M,I)
\]
This proves the statement.
\end{proof}

In particular, the restriction of the bilinear form $B$ is nondegenerate on
$\Lambda_I=d^*\Omega_N^2(M,{I^{(1)}}^\perp)\cap \Omega_D^1(M,I^{(1)})$.
For the space of all 1-forms with boundary values in $I^{(1)}$ we have:
\[
\Omega^1_D(M,I^{(1)})= \Omega^1_D(M,I^{(1)})_{cl}\oplus d^*\Omega^2_N(M,{I^{(1)}}^\perp)\cap \Omega_D^1(M,I^{(1)})
\]
The first part is the space of solutions to the Euler-Lagrange equations with
boundary values in $I^{(1)}$.

\end{document}